\documentclass[onecolumn,draftclsnofoot]{IEEEtran}
\usepackage{indentfirst}
\usepackage[dvips]{graphicx}
\usepackage{amsfonts}
\usepackage{multirow}
\usepackage{amsmath,amsthm,amssymb}
\usepackage{color,changepage}
\usepackage{algorithm}
\usepackage{algorithmic}
\usepackage{bbm}
\usepackage{verbatim}
\usepackage{makecell}
\usepackage{subfigure}
\usepackage{mathrsfs}
\usepackage{arydshln}
\usepackage{cases}
\usepackage{threeparttable}
\usepackage{extarrows}
\usepackage{array}
\usepackage{bm}
\usepackage{pifont}
\usepackage[T1]{fontenc}
\usepackage{threeparttable}
\usepackage{booktabs}
\usepackage{xr}
\usepackage{hyperref}
\usepackage{cite}
\usepackage{diagbox}
\usepackage{tabularx}
\usepackage{mathtools}
\usepackage{amsmath}
\usepackage{tikz}
\usepackage{nicematrix}
\usepackage{hhline}
\usepackage{diagbox}
\allowdisplaybreaks[4]

\definecolor{newcolor}{rgb}{0.5,0,1}

\newcolumntype{I}{!{\vrule width 2pt}}

\newtheorem{Theorem}{Theorem}

\newtheorem{Definition}{Definition}
\newtheorem{Proposition}{Proposition}

\newtheorem{Lemma}{Lemma}

\theoremstyle{remark}

\newtheorem{Example}{Example}

\DeclareMathAlphabet{\mathpzc}{OT1}{pzc}{m}{it}

\definecolor{newcolor}{rgb}{0.5,0,1}

\def\BibTeX{{\rm B\kern-.05em{\sc i\kern-.025em b}\kern-.08em
    T\kern-.1667em\lower.7ex\hbox{E}\kern-.125emX}}

\begin{document}

\title{The Capacity of Information-Theoretic Secure Aggregation in Federated Learning}
\author{Lanxin Yi, Jinbao Zhu, Kai Wan, and Xiaohu Tang
\thanks{Lanxin Yi, Jinbao Zhu, and Xiaohu Tang are 
with the Information Coding and Transmission (ICT) Key Laboratory of Sichuan Province,
Southwest Jiaotong University, Chengdu 611756, China
(email: 
lxyi@my.swjtu.edu.cn, jinbaozhu@swjtu.edu.cn, xhutang@swjtu.edu.cn).

Kai Wan is with the School of Electronic Information and Communications,
Huazhong University of Science and Technology, Wuhan 430074, China
(email: kai\_wan@hust.edu.cn).
}}
\maketitle

\begin{abstract}
Secure aggregation is a fundamental component of federated learning that enables a central server to aggregate local model updates from multiple users while preserving the privacy of each individual update without compromising model accuracy. A major bottleneck in scaling federated learning to a large number of users is the communication overhead incurred by secure aggregation, which has motivated the information theory community to study its fundamental limits. To provide privacy guarantees, existing information-theoretic secure aggregation problems typically assume that correlated random keys among users are either provided by a trusted third party (TTP) or generated according to a prescribed symmetric groupwise structure, while the communication cost required to establish such correlated randomness is often ignored. Consequently, the fundamental limits of secure aggregation under general key-distribution mechanisms remain largely unknown.
In this paper, we study the $T$-colluding information-theoretic secure aggregation problem with $N$ users and a single server under a general two-phase framework consisting of a key distribution phase and an update aggregation phase. Unlike prior works, we explicitly model the key distribution phase through user-to-user communication and allow arbitrary key distribution mechanisms established through user cooperation, thereby eliminating the need for a TTP or any prescribed groupwise key-distribution structure. This formulation enables a unified treatment of both phases and allows us to jointly characterize three fundamental resources in secure aggregation: the amount of randomness required for security, the communication required for key distribution, and the communication required for update aggregation.
We completely characterize the capacity region among these three resources. Specifically, the optimal key rate and key-distribution communication rate are both $\frac{N(N-1)}{N-T}$, while the optimal aggregation communication rate is $N$. These optimal rates are simultaneously achieved by a novel secure aggregation scheme together with a matching information-theoretic converse.
In particular, we develop an explicit deterministic capacity-achieving construction over any finite field of size at least $N$, whereas most existing schemes either rely on TTP or employ randomized or existential constructions over sufficiently large finite fields. We further show that the optimal performance can be achieved using only a pairwise key-distribution structure, enabling practical realization through standard cryptographic key-establishment mechanisms such as Diffie--Hellman key exchange. Compared with Google’s seminal secure aggregation scheme, the proposed scheme requires fewer random masking keys while preserving the same aggregation communication overhead.
\end{abstract}


\section{Introduction}\label{Introduction}
Federated learning (FL) has emerged as a prominent distributed machine learning paradigm that enables multiple participant users to collaboratively train a global model under the orchestration of a central server \cite{mcmahan2017communication,konevcny2016federated,li2020federated,yang2019federated}. 
In this framework, users perform local training using their individual data and transmit only the computed model updates (or gradients) to the server, which aggregates these updates to refine the global model and subsequently broadcasts the updated model back to the users for the next iteration. 
While FL mitigates the risk of privacy leakage by keeping raw data on local devices, recent studies have demonstrated that sharing local updates can still reveal sensitive information about the underlying training data through what are known as gradient leakage attacks \cite{fredrikson2015model,zhu2019deep,geiping2020inverting}, which allow an attacker to extract sensitive training data 
from the shared updates. Protecting the privacy of users' local updates is therefore a fundamental requirement in federated learning.

Secure aggregation was initially introduced by Bonawitz \textit{et al.} \cite{bonawitz2017practical} to enable the server to aggregate users' local updates while preserving the privacy of each individual update without compromising model accuracy. 
To achieve this goal, users first establish correlated random keys through cryptographic key-establishment mechanisms such as Diffie--Hellman key exchange \cite{katz2007introduction} and Shamir's secret sharing \cite{shamir1979share}. 
Each user then masks its local update using these correlated keys before uploading the masked update to the server. 
The masking and aggregation mechanisms are carefully designed so that the server can recover only the desired aggregate while learning no additional information about any individual update.
Since then, secure aggregation has been widely studied under a two-phase framework consisting of a key distribution phase followed by an update aggregation phase. 
In the key distribution phase, correlated random keys are established among users for masking purposes, whereas in the update aggregation phase, users upload masked updates to enable aggregation. 
This framework is particularly appealing because the key distribution phase is independent of the users' local updates and can therefore be performed offline, thereby reducing the online communication cost of masked-update aggregation to a linear order in the number of users.
It is worth noting that some works \cite{kadhe2020fastsecagg,so2021turbo,jahani2023swiftagg+} instead rely solely on locally generated independent randomness together with user-to-user communication during aggregation. 
Since such locally generated keys lack correlation across users, these approaches typically incur a quadratic online communication overhead. 
This contrast highlights the fundamental role of correlated keys in enabling communication-efficient secure aggregation.

The two-phase secure aggregation has received substantial treatment from both cryptographic and information-theoretic perspectives. 
Cryptographic secure aggregation focuses on computationally secure protocols under hardness assumptions, with significant efforts devoted to improving practical efficiency \cite{bell2020secure,choi2020communication,liu2022efficient,zheng2022aggregation,bonawitz2019federated,ergun2021sparsified,lu2023top}; see also the surveys \cite{wen2023survey,9830997,zhang2024survey}. 
In contrast, information-theoretic secure aggregation aims to characterize the minimum amount of randomness and communication required to guarantee perfect
privacy against computationally unbounded adversaries.

In this paper, we consider the information-theoretic setting and consider a traditional $T$-colluding secure aggregation problem, where the server may collude with any $T$ out of $N$ users to infer the private local updates of the remaining users. 
The work most closely related to this paper is \cite{zhao2023secure}, which studies the $T$-colluding secure aggregation problem under two key distribution scenarios.
In the first scenario, a trusted third party (TTP) distributes correlated random keys to users before aggregation. 
In the second scenario, the TTP assumption is removed by adopting a symmetric uncoded groupwise key-distribution structure, in which every subset of $G$ users among the $N$ users shares one uncoded common random key, and all such shared keys have the same entropy size.
Under these scenarios, reference \cite{zhao2023secure} characterized the minimum amount of random keys required for security and the optimal communication overhead of the update aggregation phase.
In recent years, significant progress has been made in exploring the fundamental limits of various extensions of information-theoretic secure aggregation, including secure aggregation with user dropouts \cite{zhao2022information,so2022lightsecagg,wan2024information,wan2024capacity,zhang2025securegroup}, hierarchical secure aggregation \cite{zhang2025optimal,li2026fundamental}, vector secure aggregation \cite{yuan2025vector,hu2026capacity}, decentralized secure aggregation \cite{Zhang2026Decentralized,li2025capacity}, and weak privacy constraints \cite{li2025weakly,li2025hierarchical}. 

In general, the information-theoretic secure aggregation problem has been extensively studied, and its fundamental limits have been well understood under various settings. Nevertheless, several fundamental questions remain open:
\begin{itemize}
\item  Existing works \cite{zhao2023secure,zhao2022information,zhang2025optimal,li2026fundamental,yuan2025vector,hu2026capacity,Zhang2026Decentralized,li2025weakly,li2025hierarchical} 
assume the existence of a TTP that distributes correlated random keys to users before aggregation. This assumption, while convenient for analyzing fundamental limits, is impractical in real federated learning systems due to its reliance on a fully trusted external entity. Some studies have removed the TTP assumption by adopting symmetric groupwise key structures or Shamir's secret sharing \cite{zhao2023secure,so2022lightsecagg,wan2024information,wan2024capacity,zhang2025securegroup,li2025capacity}. However, these approaches impose highly structured forms of correlated randomness, leaving it unclear whether more general key distribution mechanisms could reduce the required randomness or achieve information-theoretic optimality.
\item Existing works \cite{li2025hierarchical,li2025weakly,Zhang2026Decentralized,li2025capacity,hu2026capacity,zhao2023secure,so2022lightsecagg,zhang2025optimal,yuan2025vector,zhao2022information,wan2024information,wan2024capacity,zhang2025securegroup,li2026fundamental} mainly focus on the communication overhead of the update aggregation phase, while the communication cost required to establish the correlated keys is typically not explicitly characterized. 
When the correlated keys are generated through user cooperation rather than assumed a priori, the fundamental communication cost of the key distribution phase remains unknown.
\item Most existing secure aggregation problems lack explicit deterministic constructions over small finite fields. Except for the $T$-colluding secure aggregation problem with user dropouts in \cite{so2022lightsecagg}, 
the achievable schemes in the aforementioned works either rely on the TTP assumption \cite{zhao2023secure,zhao2022information,zhang2025optimal,li2026fundamental,yuan2025vector,hu2026capacity,Zhang2026Decentralized,li2025weakly,li2025hierarchical}, or on randomized constructions over sufficiently large finite fields \cite{zhao2023secure,wan2024information,wan2024capacity,zhang2025securegroup,li2025capacity,li2025weakly,li2025hierarchical}, or on existential constructions over sufficiently large finite fields \cite{zhang2025optimal,li2026fundamental}. 
Developing explicit deterministic constructions over small finite fields therefore remains an important open problem.
\end{itemize}

Motivated by these questions, we revisit the $T$-colluding information-theoretic secure aggregation problem under a general two-phase formulation. 
Instead of imposing any prescribed key distribution structure, we explicitly model the key distribution phase through user-to-user communication. 
Each user locally generates private randomness and communicates encoded key symbols to other users before the aggregation phase. 
The communication required for establishing the correlated keys is explicitly counted together with the communication cost of masked-update aggregation, as well as the amount of random keys required for security. 
This formulation eliminates the TTP assumption while allowing arbitrary key-distribution mechanisms among users. 
It further enables us to jointly characterize three fundamental resources in secure aggregation: the amount of random keys required for security, the communication required for key establishment, and the communication required for update aggregation.
The corresponding normalized quantities are referred to as the key rate $R_Z$, the key-distribution communication rate $R_K$, and the aggregation communication rate $R_A$, respectively.

As a result, the main contributions of this paper are summarized as follows.
 We formulate a two-phase $T$-colluding information-theoretic secure aggregation problem in which the key distribution phase and the update aggregation phase are jointly accounted for. 
In particular, we completely characterize the capacity region, i.e., the set of all achievable rate tuples $(R_Z,R_K,R_A)$, 
as $\left\{(R_Z, R_K, R_A)\in\mathbb{R}^3:R_Z\geq \frac{N (N-1)}{N - T},R_K \geq \frac{N (N-1)}{N - T},R_A\geq N \right\}$.
Moreover, we show that all three lower bounds are tight and are simultaneously achieved by an explicit secure aggregation scheme over any finite field $\mathbb{F}_q$ with $q\geq N$.
More precisely,
\begin{enumerate}
\item We propose an explicit capacity-achieving secure aggregation scheme over any finite field $\mathbb{F}_q$ with $q\geq N$. 
The proposed scheme is built upon a general linear coding framework together with carefully designed encoding matrices 
that exploit the algebraic structure of
Vandermonde matrices.
In contrast to the prior capacity-achieving scheme \cite{zhao2023secure}, which relies on randomized constructions over sufficiently large finite fields, the proposed scheme provides a deterministic explicit construction over a finite field whose size grows only linearly with the number of users.

\item We develop a matching information-theoretic converse for the general two-phase secure aggregation formulation. 
Unlike previous converse results derived under the TTP or prescribed groupwise key-distribution assumptions \cite{zhao2023secure}, the proposed converse applies to arbitrary key-distribution mechanisms among users. 
The converse is established by first characterizing the effective randomness that can contribute to masking users' local updates.
This characterization is then combined with the security constraint to show that, for every colluding set, the encoded keys exchanged among the non-colluding users must collectively contain sufficient randomness to conceal their individual updates beyond the desired aggregation result.
By summing the resulting entropy inequalities over all colluding sets, we obtain tight lower bounds on both the randomness requirement and the key-distribution communication cost, while the aggregation communication bound follows directly from the decodability constraint.


\item The proposed capacity-achieving scheme admits an implementation based on pairwise random key distribution, in which each pair of users shares a common random key and all such shared keys are mutually independent. 
This establishes that a pairwise key-distribution structure is sufficient for achieving the optimal performance of secure aggregation. Furthermore, when the information-theoretic security requirement is relaxed to computational security, the proposed scheme can be efficiently implemented using practical cryptographic key-establishment mechanisms such as the Diffie--Hellman key exchange protocol. Compared with Google's original secure aggregation scheme in \cite{bonawitz2017practical}, the proposed scheme requires fewer pairwise masking keys while preserving the same communication overhead in the aggregation phase. 
\end{enumerate}

The remainder of this paper is organized as follows.
Section \ref{Problem:formulation} formulates the two-phase information-theoretic secure aggregation problem.
Section \ref{Main Result} describes the main results of this paper and provides intuitive insights into them.
Section \ref{achievability Proof} develops a general linear coding framework that achieves the optimal performance for the secure aggregation problem and presents its implementation.
Section \ref{Converse Proof} establishes a matching information-theoretic converse bound. Finally, this paper is concluded in Section \ref{conclusion}.

\subsubsection*{Notation}
For any positive integer $n$, let $[n]$ denote the set $\{1,2,\ldots,n\}$. 
Let cursive capital letters denote sets, such as $\mathcal{A}$, while $|\mathcal{A}|$ denotes its cardinality. For any two sets $\mathcal{A}$ and $\mathcal{B}$, $\mathcal{A}\backslash\mathcal{B}$ denotes the set of elements in $\mathcal{A}$ but not in $\mathcal{B}$, i.e., $\mathcal{A}\backslash\mathcal{B} \triangleq \{x \in \mathcal{A} : x \notin \mathcal{B}\}$.
Let boldface capital letters represent matrices. 
For any nonnegative integer $n$, we adopt the convention for binomial coefficients that $\binom{n}{m}=0$ for $m<0$ and $\binom{n}{0}=1$.
For any positive integers $m$ and $n$, we use $\mathbf{I}_n$ to denote the $n\times n$ identity matrix, and $\mathbf{0}_{m\times n}$ to denote the $m\times n$ all-zero matrix. For any matrix $\mathbf{A}$, let $\operatorname{rank}(\mathbf{A})$ denote its rank. Let $\otimes$ denote the Kronecker product between two matrices.  
For any block matrix $\mathbf{A}=(\mathbf{A}_{i,j})_{i\in[m],j\in[n]}$, let $\mathbf{A}_{\mathcal{S}_1,\mathcal{S}_2}=(\mathbf{A}_{i,j})_{i\in\mathcal{S}_1,j\in\mathcal{S}_2}$ denote the block submatrix of $\mathbf{A}$ obtained by extracting the row blocks indexed by $\mathcal{S}_1$ and the column blocks indexed by $\mathcal{S}_2$ for any two subsets $\mathcal{S}_1\subseteq[m]$ and $\mathcal{S}_2\subseteq[n]$.
For any two subsets $\mathcal{S}_1=\{m_1,m_2,\ldots,m_{s_1}\}\subseteq[m]$ and $\mathcal{S}_2=\{n_1,n_2,\ldots,n_{s_2}\}\subseteq[n]$, let $\mathrm{diag}(\mathbf{A}_{\mathcal{S}_1,\mathcal{S}_2})$ denote the matrix formed by the collection of the matrix blocks $\{\mathbf{A}_{i,j}\}_{i\in\mathcal{S}_1,j\in\mathcal{S}_2}$, given by
\begin{IEEEeqnarray}{c}
\mathrm{diag}(\mathbf{A}_{\mathcal{S}_1,\mathcal{S}_2})=
\begin{bmatrix}
\mathrm{diag}(\mathbf{A}_{m_1,n_{1}}, \mathbf{A}_{m_2,n_{1}}, \cdots, \mathbf{A}_{m_{s_1},n_{1}} )\\
\mathrm{diag}(\mathbf{A}_{m_1,n_{2}}, \mathbf{A}_{m_2,n_{2}}, \cdots, \mathbf{A}_{m_{s_1},n_{2}} )\\
\vdots\\
\mathrm{diag}(\mathbf{A}_{m_1,n_{s_2}}, \mathbf{A}_{m_2,n_{s_2}}, \cdots, \mathbf{A}_{m_{s_1},n_{s_2}} )\\
\end{bmatrix}, \notag 
\end{IEEEeqnarray}
where for any given $j\in[s_2]$, $\mathrm{diag}(\mathbf{A}_{m_1,n_{j}}, \mathbf{A}_{m_2,n_{j}}, \cdots, \mathbf{A}_{m_{s_1},n_{j}} )$ is a block-diagonal matrix whose diagonal blocks are given by $\{\mathbf{A}_{m_i,n_{j}}\}_{i\in[s_1]}$, i.e.,
\begin{IEEEeqnarray}{c}
\mathrm{diag}(\mathbf{A}_{m_1,n_{j}}, \mathbf{A}_{m_2,n_{j}}, \cdots, \mathbf{A}_{m_{s_1},n_{j}} )=
\begin{bmatrix}
    \mathbf{A}_{m_1,n_{j}}  &  & &  \\
      & \mathbf{A}_{m_2,n_{j}}  & &  \\
     &  & \ddots &     \\
    &  & & \mathbf{A}_{m_{s_1},n_{j}}
\end{bmatrix}. \notag
\end{IEEEeqnarray} 

\section{Problem Formulation}\label{Problem:formulation}
Consider a federated learning system with one server and $N$ users, each holding a confidential local input (e.g., a model update or gradient), where the users wish to securely aggregate all local inputs at the server. 
We consider information-theoretic security for the users' inputs under an honest-but-curious setting. In this setting, both the server and the users honestly follow the prescribed protocol, while the server colluding with any subset of up to $T$ users may attempt to infer information about the local inputs of the remaining users.
We assume the existence of private channels among all users and between each user and the server, ensuring that communication within the federated learning system can be carried out securely.\footnote{It is well known that this is a very natural assumption for achieving information-theoretic security guarantees. In practice, such private channels can be established, for example, through public-key encryption techniques \cite{bonawitz2017practical,katz2007introduction} or quantum key distribution technologies \cite{scarani2009security}.}


The local input of each user $n\in[N]$ can be represented by a vector $W_n$ of length $L$, whose elements are drawn independently and uniformly from a finite field $\mathbb{F}_q$ for a prime power $q$. All input vectors across users are mutually independent,\footnote{The uniformity and independence of the input vectors are essential for the converse proof, but are not necessary for establishing achievability.} i.e.,
\begin{IEEEeqnarray}{rCl}
H(W_1,\ldots,W_N)=\sum_{n \in [N]} H(W_n), \label{model:file:inden} \\
H(W_1)=\ldots=H(W_N)=L. \label{infor:indenpe}
\end{IEEEeqnarray}

Let $W_{sum}$ represent the aggregation of all users' input vectors, i.e., 
\begin{IEEEeqnarray}{c}
W_{sum}\triangleq\sum_{n\in[N]}W_n.
\notag
\end{IEEEeqnarray}
To complete aggregation in the federated learning system, the secure aggregation problem is divided into two phases: \emph{the key distribution phase} and \emph{the input aggregation phase}. In the key distribution phase, users exchange a set of correlated random keys to mask their inputs. In the input aggregation phase, the masked inputs are communicated in order to perform the desired aggregation. We describe these two phases precisely as follows.
\begin{itemize}
    \item \emph{Key distribution phase:} The user $n\in[N]$ locally generates a random key $Z_n$ and then sends to another user $m$ an encoded version $Z_{n,m}$ of the random key, i.e.,\footnote{Our formulated key distribution model includes the widely adopted symmetric groupwise key-distribution structure \cite{zhao2023secure,wan2024information,wan2024capacity,zhang2025securegroup,li2025capacity} as a special case. More specifically, in the symmetric groupwise key-distribution structure, for any subset of users $\mathcal{G}\subseteq[N]$ with cardinality $2\leq G\leq N-T$, there exists a random key $Z_{\mathcal{G}}$ shared among all users in $\mathcal{G}$. Such a shared random key $Z_{\mathcal{G}}$ can be realized within our key distribution framework by allowing an arbitrary user $n\in\mathcal{G}$ to locally generate $Z_{\mathcal{G}}$ and distribute it to the remaining users $m\in\mathcal{G}\backslash\{n\}$.}
    \begin{IEEEeqnarray}{c}\label{encodedkey}
        H(Z_{n,m}|Z_n)=0,\quad \forall\,n\in[N],m\in[N]\backslash\{n\}.
    \end{IEEEeqnarray}
    \item \emph{Input aggregation phase:} Upon receiving the encoded keys 
$\{Z_{m,n}\}_{m\in[N]\backslash\{n\}}$ from the other users, user $n$ generates a masked input $X_n$ and sends it to the server. The masked input is a deterministic function of the local input $W_n$, the local key $Z_n$, and the received encoded keys $\{Z_{m,n}\}_{m\in[N]\backslash\{n\}}$, i.e.,
\begin{IEEEeqnarray}{c}
H(X_n|W_n,Z_n, \{Z_{m,n}\}_{m \in [N]\backslash\{n\}} ) = 0, \quad \forall\, n \in [N].
\label{model:maskedinput}
\end{IEEEeqnarray}
The server recovers the desired aggregation $W_{sum}$ from the masked inputs $\{X_n\}_{n\in[N]}$ collected from all users.
\end{itemize}

The following information-theoretic constraints must be satisfied by any secure aggregation scheme.
\begin{itemize}
    \item \textbf{Correctness:} The desired aggregation $W_{sum}$ must be fully determined by the masked inputs $\{X_n\}_{n\in[N]}$, i.e.,
    \begin{IEEEeqnarray}{c}\label{model:correctness}
    H(W_{sum} |\{X_n\}_{n\in[N]} ) = 0.
    \end{IEEEeqnarray}
    \item \textbf{Security:} 
    For any subset $\mathcal{T} \subseteq [N]$ of up to $T$ colluding users, the security requirement ensures that the server together with the colluding users $\mathcal{T}$ must not learn any additional information about the input vectors $\{W_n\}_{n \in [N]}$ beyond what is revealed by the aggregation itself, i.e., the aggregation result $W_{sum}$ and the inputs of the colluding users $\{W_n\}_{n\in\mathcal{T}}$. 
    Formally, the following security constraint must hold for any subset $\mathcal{T} \subseteq [N]$  
    with $|\mathcal{T}| \le T$:
    \begin{IEEEeqnarray}{c}\label{model:security}
     I(\{W_n\}_{n\in[N]}; \{X_n\}_{n\in[N]}, \{Z_n\}_{n\in \mathcal{T}}, \{Z_{m,n}\}_{m\in[N]\backslash\{n\},n\in\mathcal{T}}|W_{sum},\{W_n\}_{n\in\mathcal{T}}) = 0,
    \end{IEEEeqnarray}
    where the collection $\{\{X_n\}_{n\in[N]}, \{Z_n\}_{n\in \mathcal{T}}, \{Z_{m,n}\}_{m\in[N]\backslash\{n\},n\in\mathcal{T}}\}$ represents all the information accessible to the server together with the colluding users. Notably, from the perspective of the server together with the colluding users, the observed data $\{Z_{n,m}\}_{n\in\mathcal{T},m\in[N]\backslash\{n\}}$ 
    is fully determined by the random keys $\{Z_n\}_{n\in\mathcal{T}}$ by \eqref{encodedkey}, and hence can be ignored.
\end{itemize}

Since the users' local random keys are generated independently of their input vectors, it naturally follows that the independence among the relevant quantities can be specified as follows:
\begin{IEEEeqnarray}{c}\label{model:key:inden}
H(W_1,\ldots,W_N,Z_1,\ldots,Z_N)=\sum_{n \in [N]} H(W_n)+\sum_{n \in [N]} H(Z_n).
\end{IEEEeqnarray}

The performance of a secure aggregation scheme is evaluated using the following three metrics.
\begin{enumerate}
\item[1.] The key rate $R_Z$, which measures the number of generated key symbols required per desired aggregation symbol to guarantee information-theoretic security in the federated learning system, is defined as
   \begin{IEEEeqnarray}{c}
     R_Z \triangleq\frac{\sum_{n\in[N]}H(Z_n)}{H(W_{sum})}=\frac{L_Z}{L},\label{definition:RZ}
    \end{IEEEeqnarray}
    where $L_Z=\sum_{n\in[N]}H(Z_n)$ denotes the total number of key symbols used in the federated learning system.
  \item[2.] The key-distribution communication rate $R_K$ and the aggregation communication rate $R_A$, which quantify the number of symbols that need to be transmitted per desired aggregation symbol during the key distribution phase and the input aggregation phase, respectively, are defined as 
    \begin{IEEEeqnarray}{c}  
    R_K\triangleq\frac{\sum_{n\in[N]}\sum_{m\in[N]\backslash\{n\}}H(Z_{n,m})}{H(W_{sum})}=\frac{L_K}{L}, \quad
    R_A\triangleq\frac{\sum_{n\in[N]}H(X_n)}{H(W_{sum})}= \frac{L_A}{L}, \label{definition:RKRA}
    \end{IEEEeqnarray}
where $L_K=\sum_{n\in[N]}\sum_{m\in[N]\backslash\{n\}}H(Z_{n,m})$ and $L_A=\sum_{n\in[N]}H(X_n)$ denote the total communication costs in the key distribution phase and the input aggregation phase, respectively.
  \item[3.] The finite field size $q$, which specifies the alphabet size and reflects the implementation complexity of the secure aggregation scheme.
\end{enumerate}

\begin{Definition}[Capacity Region of Secure Aggregation]\label{def_capacity}
The rate tuple $(R_Z, R_K, R_A)$ is said to be \emph{achievable} if there exists a secure aggregation scheme that satisfies both the correctness constraint \eqref{model:correctness} and the security constraint \eqref{model:security}, and whose key rate and communication rates do not exceed the given values $R_Z, R_K$, and $R_A$. Furthermore, the \emph{capacity region} of the secure aggregation problem is defined as the closure of the set of all achievable rate tuples, denoted by $\mathcal{R}^*$, i.e.,
\begin{IEEEeqnarray}{c}
\mathcal{R}^*=\{(R_Z, R_K, R_A)\in\mathbb{R}^{3}:(R_Z, R_K, R_A) \text{\,\,is achievable}\}.\notag
\end{IEEEeqnarray}
\end{Definition}

The objective of this paper is to characterize the capacity region $\mathcal{R}^{*}$ of the secure aggregation problem and to provide explicit capacity-achieving constructions over a finite field of the smallest possible size. 

\section{Main Results}\label{Main Result}
Under the proposed two-phase secure aggregation formulation, the main results of this paper are twofold. First, we completely characterize the capacity region in terms of the key rate $R_Z$, the key-distribution communication rate $R_K$, and the aggregation communication rate $R_A$. Second, we develop an explicit capacity-achieving secure aggregation scheme over a finite field whose size grows linearly with the number of users, and further show that the optimal performance can be achieved using only a pairwise key-distribution structure. The achievability proof is presented in Section~\ref{achievability Proof}, while the converse proof is provided in Section~\ref{Converse Proof}. Throughout this paper, we focus on the nontrivial parameter regime $N>1$ and $0\leq T<N-1$. This is because when $T=N-1$ or $T=N$, the server colluding with any $T$ users can recover all users' local inputs from the desired aggregation result and the colluding users' inputs, and hence the security constraint becomes vacuous.
\begin{Theorem}\label{theo:capacity}
For the federated learning system with $N>1$ users and at most $0\leq T<N-1$ colluding users, the capacity region of the secure aggregation problem is given by 
\begin{IEEEeqnarray}{c}
\mathcal{R}^*=\left\{(R_Z, R_K, R_A)\in\mathbb{R}^3:R_Z\geq \frac{N (N-1)}{N - T},R_K \geq \frac{N (N-1)}{N - T},R_A\geq N \right\}.
\notag
\end{IEEEeqnarray}
\end{Theorem}
Based on this theorem, we can make the following interesting observations and explanations, which provide valuable guidance for designing a secure aggregation scheme that achieves the optimal rate tuple $\big(\frac{N (N-1)}{N - T},\frac{N (N-1)}{N - T},N\big)$:
\begin{itemize}
    \item The optimal key-distribution communication rate $R_K$ is equal to the optimal key rate $R_Z$. This implies that, in order to design a secure aggregation scheme achieving the optimal rate tuple, it is preferable that the local random key $Z_n$ generated by each user $n\in[N]$ should be fully distributed to the other $N-1$ users during the key distribution phase, such that the size of the local key $Z_n$ is equal to the total size of the encoding keys $\{Z_{n,m}\}_{m\in[N]\backslash\{n\}}$ sent to the other users.
    \item The optimal key-distribution communication rate $R_K$ and the optimal key rate $R_Z$ both increase with the security parameter $T$. This phenomenon is natural, since a larger value of $T$ requires more random keys to provide stronger security guarantees, which in turn necessitates additional communication to distribute these extra random keys to other users.
    However, the optimal aggregation communication rate $R_A$ is independent of the security parameter $T$. This indicates that for different security parameters $T$, a secure aggregation scheme achieving the optimal rate tuple cannot satisfy the required correctness and security constraints by increasing the communication in the aggregation phase.
    Instead, it requires a carefully designed key distribution mechanism together with an appropriate key encoding strategy.
    \item The optimal aggregation communication rate $R_A$ is $N$, which implies that the masked input $X_n$ generated by each user $n\in[N]$ must have the same length as its local input $W_n$. In particular, after eliminating the random masks used in the construction of $X_n$, the underlying input $W_n$ must be perfectly recoverable; otherwise, the desired aggregation cannot be accomplished.
    \item In the case of $T=0$, i.e., when there are no colluding users, the optimal key rate $R_Z$ is $N-1$. This is because the $N$ local inputs of all users must be perfectly protected except for their aggregated result, which requires $N-1$ independent random keys of normalized size. When $T>0$, i.e., the server colludes with any $T$ users, the optimal key rate increases by a factor of $\frac{N}{N-T}$ compared to the case $T=0$. More specifically, for $T=0$, since all users are non-colluding and their local random keys are mutually independent, each user possesses $N$-dimensional independent encoding keys, formed by one locally generated component and $N-1$ components received from other users, which collectively mask its local input to provide information-theoretic security. 
    However, when $T>0$ users collude with the server, among the $N$-dimensional encoding keys available at each non-colluding user, $T$ dimensions are known to the colluding users and hence cannot be used to provide secrecy. 
    Only the remaining $N-T$ dimensions can be used to mask the local input, which is insufficient to guarantee security. To compensate for this loss, the remaining $N-T$ dimensions of usable encoding keys must be expanded by a factor of $\frac{N}{N-T}$ so that effective $N$-dimensional encoding keys can be utilized, as in the case $T=0$. Consequently, the key rate increases by a factor of $\frac{N}{N-T}$ compared to the case $T=0$.
\end{itemize}

\begin{Theorem}\label{theo:achive}
For the federated learning system with $N>1$ users and at most $0\leq T<N-1$ colluding users, there exists a secure aggregation scheme that achieves the optimal rate tuple $\big(\frac{N (N-1)}{N - T},\frac{N (N-1)}{N - T},N\big)$ over any finite field $\mathbb{F}_q$ with size $q\geq N$. 
\end{Theorem}
The achievability of Theorems \ref{theo:capacity} and \ref{theo:achive} can be established by constructing a secure aggregation scheme that achieves the optimal rate tuple $\big(\frac{N (N-1)}{N - T},\frac{N (N-1)}{N - T},N\big)$ over any finite field $\mathbb{F}_q$ with size $q\geq N$. 
This is completed by developing a general linear coding framework for the secure aggregation problem based on a collection of encoding matrices. 
This framework provides a systematic approach for designing secure aggregation schemes that achieve the optimal performance while facilitating the reduction of the finite field size. In particular, within this framework, we derive sufficient and necessary conditions on the encoding matrices to ensure the achievability of the coding framework. By leveraging the structures inspired by Vandermonde matrices and full-rank square matrices, we explicitly construct a set of encoding matrices that satisfy these conditions over any finite field $\mathbb{F}_q$ of size $q\geq N$.
Consequently, by instantiating the proposed coding framework with the constructed encoding matrices, we obtain an explicit secure aggregation scheme that achieves the optimal performance over an arbitrary finite field $\mathbb{F}_q$ of size $q\geq N$, thereby completing the achievability proof of Theorems \ref{theo:capacity} and \ref{theo:achive}.


\begin{Theorem}\label{theorem:pairwise}
For the $T$-colluding information-theoretic secure aggregation problem with $N>1$ and $0\leq T<N-1$, the optimal rate tuple $\big(\frac{N (N-1)}{N - T},\frac{N (N-1)}{N - T},N\big)$ can be achieved using a pairwise key-distribution structure over any finite field $\mathbb{F}_q$ with $q\geq N$, in which every pair of users shares a mutually independent common random key.
\end{Theorem}
This theorem follows by demonstrating that the proposed capacity-achieving secure aggregation scheme admits an implementation based on pairwise key distribution. When the information-theoretic security requirement is relaxed to computational security, the pairwise key distribution structure can be established via standard cryptographic key exchange protocols, such as the Diffie--Hellman protocol, thereby ensuring the practical feasibility of the proposed scheme.
Pairwise-key-based secure aggregation was first introduced in the seminal work \cite{bonawitz2017practical} from the perspective of computational security, and was subsequently investigated in \cite{zhao2023secure} under the information-theoretic security setting. Compared with the most closely related work \cite{zhao2023secure}, which relies on randomized constructions over finite fields with asymptotically infinite size and requires the input length $L$ to grow super-exponentially with the number of users $N$, we develop an explicit deterministic construction that requires only a finite field size $q\geq N$ and an input length $L=N-T$. 
Furthermore, compared with the conventional scheme in \cite{bonawitz2017practical}, the proposed scheme achieves a smaller key rate $R_Z$ while preserving the same communication rate $R_A$ in the aggregation phase. A detailed comparison is provided in Section \ref{section:Comparison}.

\section{Proof of Theorems \ref{theo:capacity}--\ref{theorem:pairwise}: Achievability}\label{achievability Proof}

In this section, we first develop a capacity-achieving linear coding framework for the secure aggregation problem and then realize it by constructing the corresponding encoding matrices over any finite field $\mathbb{F}_q$ with size $q\geq N$.
In addition, we establish that the proposed capacity-achieving scheme admits an implementation based on pairwise key distribution. These results jointly complete the achievability proofs of Theorems \ref{theo:capacity}--\ref{theorem:pairwise}.
Finally, we compare the proposed scheme with existing related works and highlight its advantages.

\subsection{Capacity-Achieving Linear Coding Framework for Secure Aggregation}
Assume that $L$ is a multiple of $N-T$. Note that, in practice, this condition can be easily satisfied by padding the local inputs with a small number of zeros.
Without loss of generality, assume that the local input $W_n$ of each user $n\in[N]$ is viewed as a column vector of dimension $L$.
Next, we formally present a general capacity-achieving linear coding framework by describing the encoding and communication process of the secure aggregation problem. 

In the key distribution phase, each user $n\in[N]$ generates a local random key $Z_n$ by independently and uniformly selecting $\frac{N-1}{N-T}L$ symbols from the finite field $\mathbb{F}_q$. For convenience, the local random key $Z_n$ is represented as a column vector of dimension $\frac{N-1}{N-T}L$. Let $\mathbf{G}_{n,m}$ denote an encoding matrix of dimensions $\frac{L}{N-T}\times \frac{N-1}{N-T}L$ over $\mathbb{F}_q$, which is used to encode the random key $Z_{n}$ and generate the shared encoding key $Z_{n,m}$ for any $n\in[N]$ and $m\in[N]\backslash\{n\}$. Accordingly, the encoded key $Z_{n,m}$ sent from user $n$ to user $m$ is given by
\begin{IEEEeqnarray}{c}
    Z_{n,m}=\mathbf{G}_{n,m}Z_n,\quad \forall\, n\in [N],m\in[N]\backslash\{n\}. \label{eq:matrixofcodedkey}
\end{IEEEeqnarray}
Remarkably, the dimensions of $\{Z_n\}_{n\in[N]}$ and $\{\mathbf{G}_{n,m}\}_{n\in[N],m\in[N]\backslash\{n\}}$ are chosen such that the resulting coding framework achieves both the optimal key rate and the optimal key-distribution communication rate.


In the input aggregation phase, each user $n\in[N]$ generates the masked input $X_n$ by encoding the currently available data $W_n$, $Z_n$, and $\{Z_{m,n}\}_{m \in [N]\backslash\{n\}}$. More precisely, let $\mathbf{E}_{n,n}$ denote an encoding matrix of dimensions $L\times\frac{N-1}{N-T}L$ over $\mathbb{F}_q$, which is used to encode the local random key $Z_{n}$ for any $n\in[N]$. Let $\mathbf{P}_{m,n}$ denote an encoding matrix of dimensions $L\times\frac{L}{N-T}$ over $\mathbb{F}_q$, which is used to encode the received encoding key $Z_{m,n}$ for any $n\in[N]$ and $m\in[N]\backslash\{n\}$. Then, the masked input $X_n$ is given by
\begin{IEEEeqnarray}{c}
X_n=W_n+\mathbf{E}_{n,n}Z_{n}+\sum_{m\in[N]\backslash\{n\}}\mathbf{P}_{m,n}Z_{m,n}, \quad\forall\,n\in[N]. \label{eq:maskinputlinear}
\end{IEEEeqnarray}
Similarly, the dimensions of the encoding matrices $\{\mathbf{E}_{n,n}\}_{n\in[N]}$ and $\{\mathbf{P}_{m,n}\}_{m\in[N]\backslash\{n\},n\in[N]}$ are chosen to guarantee that the coding framework achieves the optimal aggregation communication rate.

By combining \eqref{eq:matrixofcodedkey} and \eqref{eq:maskinputlinear}, the masked input $X_n$ can be equivalently expressed as
\begin{IEEEeqnarray}{c}
X_n=W_n+\sum_{m\in[N]}\mathbf{E}_{m,n}Z_{m},\quad\forall\,n\in[N], \label{eq:encodingscheme}
\end{IEEEeqnarray}
where for any given $m\in[N]\backslash\{n\}$, the matrix $\mathbf{E}_{m,n}$ is of dimensions $L\times \frac{N-1}{N-T}L$, given by
\begin{IEEEeqnarray}{c}
\mathbf{E}_{m,n}\triangleq\mathbf{P}_{m,n}\mathbf{G}_{m,n},\quad\forall\, m\in[N]\backslash\{n\}. \label{framework:matrix:equation}
\end{IEEEeqnarray}  
Furthermore, the masked inputs $\{X_n\}_{n\in[N]}$ of all users can be written in the following matrix form:
\begin{IEEEeqnarray}{c}
\begin{bmatrix}
X_1\\
X_2\\
\vdots\\
X_N
\end{bmatrix}
=
\begin{bmatrix}
W_1\\
W_2\\
\vdots\\
W_N
\end{bmatrix}
+
\begin{bmatrix}
\mathbf{E}_{1,1} & \mathbf{E}_{2,1} & \cdots & \mathbf{E}_{N,1} \\
\mathbf{E}_{1,2} & \mathbf{E}_{2,2} & \cdots & \mathbf{E}_{N,2} \\
\vdots & \vdots & \ddots & \vdots \\
\mathbf{E}_{1,N} & \mathbf{E}_{2,N} & \cdots & \mathbf{E}_{N,N}
\end{bmatrix}
\begin{bmatrix}
Z_1\\
Z_2\\
\vdots\\
Z_N
\end{bmatrix}.
\label{eq:maskinput}
\end{IEEEeqnarray}

In general, for the considered secure aggregation problem, the constructed coding framework is fully described by the encoding matrices 
$\{\mathbf{E}_{m,n}\}_{m\in[N],n\in[N]}$ and $\{\mathbf{P}_{m,n},\mathbf{G}_{m,n}\}_{m\in[N]\backslash\{n\},n\in[N]}$. 
It can be directly proved that the performance of the coding framework matches the optimal rate tuple $\big(\frac{N (N-1)}{N - T},\frac{N (N-1)}{N - T},N\big)$ of the secure aggregation problem. Therefore, it suffices to focus on designing appropriate encoding matrices $\{\mathbf{E}_{m,n}\}_{m\in[N],n\in[N]}$ and $\{\mathbf{P}_{m,n},\mathbf{G}_{m,n}\}_{m\in[N]\backslash\{n\},n\in[N]}$ over a finite field of small size such that the coding framework satisfies both the correctness constraint in \eqref{model:correctness} and the security constraint in \eqref{model:security}. 

The following proposition characterizes the sufficient and necessary conditions on the encoding matrices under which the coding framework is achievable. 
For convenience, 
we define $\mathbf{E}$ as the global encoding matrix in \eqref{eq:maskinput}, 
given by
\begin{IEEEeqnarray*}{c} 
\mathbf{E}\triangleq
\begin{bmatrix}
\mathbf{E}_{1,1} & \mathbf{E}_{2,1} & \cdots & \mathbf{E}_{N,1} \\
\mathbf{E}_{1,2} & \mathbf{E}_{2,2} & \cdots & \mathbf{E}_{N,2} \\
\vdots & \vdots & \ddots & \vdots \\
\mathbf{E}_{1,N} & \mathbf{E}_{2,N} & \cdots & \mathbf{E}_{N,N}
\end{bmatrix}\in \mathbb{F}^{NL\times \frac{N(N-1)}{N-T}L}_q.
\end{IEEEeqnarray*}
\begin{Proposition}[Sufficient and Necessary Conditions for Encoding Matrices]
\label{pro:condition}
For the secure aggregation problem with system parameters $N>1$ and $0\leq T<N-1$, the constructed coding framework is an achievable scheme that achieves the optimal rate tuple $\big(\frac{N (N-1)}{N - T},\frac{N (N-1)}{N - T},N\big)$ 
\emph{if} the encoding matrices 
$\{\mathbf{E}_{m,n}\}_{m\in[N],n\in[N]}$ and $\{\mathbf{P}_{m,n},\mathbf{G}_{m,n}\}_{m\in[N]\backslash\{n\},n\in[N]}$, subject to the structural constraint
\begin{IEEEeqnarray}{c}
\mathbf{E}_{m,n}=\mathbf{P}_{m,n}\mathbf{G}_{m,n}, \quad\forall\,n\in[N],m\in[N]\backslash\{n\}, \label{eq:structure:cons}
\end{IEEEeqnarray}
satisfy the following two conditions over some finite field $\mathbb{F}_q$. Here, the matrices $\mathbf{E}_{m,n},\mathbf{P}_{m,n}$, and $\mathbf{G}_{m,n}$ have dimensions $L\times \frac{N-1}{N-T}L$, $L\times \frac{L}{N-T}$, and
$\frac{L}{N-T}\times\frac{N-1}{N-T}L$, respectively.
\begin{itemize}
\item 
The sum of the matrices $\{\mathbf{E}_{m,n}\}_{n\in[N]}$ equals the zero matrix for all $m\in[N]$, i.e.,
\begin{IEEEeqnarray}{c}
\sum_{n=1}^N \mathbf{E}_{m,n}  = \mathbf{0}_{L\times \frac{N-1}{N-T}L},\quad\forall\, m\in[N].\label{eq:zerosum-constraint}
\end{IEEEeqnarray}
This condition ensures that the coding framework satisfies the correctness constraint in \eqref{model:correctness}.
\item For any subset of colluding users $\mathcal{T}\subseteq[N]$ with size $|\mathcal{T}|\leq T$ and its complement $\mathcal{T}^c=[N]\backslash\mathcal{T}$, 
there exists some subset $\widetilde{\mathcal{T}}^c\subseteq\mathcal{T}^c$ with size $N-|\mathcal{T}|-1$, such that the matrix $\mathbf{E}_{\widetilde{\mathcal{T}}^c,\mathcal{T}^c}$ of dimensions $(N-|\mathcal{T}|-1)L\times\frac{(N-|\mathcal{T}|)(N-1)}{N-T}L$ and the matrix $\mathrm{diag}(\mathbf{G}_{\mathcal{T}^c,\mathcal{T}})$ of dimensions $\frac{|\mathcal{T}|(N-|\mathcal{T}|)}{N-T}L\times\frac{(N-|\mathcal{T}|)(N-1)}{N-T}L$ satisfy the condition:
\begin{IEEEeqnarray}{c}
\renewcommand{\arraystretch}{0.5}
\operatorname{rank}\left(
\bigg[
  \begin{array}{@{}c@{}}
\mathbf{E}_{\widetilde{\mathcal{T}}^c,\mathcal{T}^c} \\
\mathrm{diag}(\mathbf{G}_{\mathcal{T}^c,\mathcal{T}})\\
\end{array}
\bigg]
\right)
 - \operatorname{rank}\left(\mathrm{diag}(\mathbf{G}_{\mathcal{T}^c,\mathcal{T}})\right) =(N-|\mathcal{T}|-1) L. \label{eq:rankcondition}
\end{IEEEeqnarray}
This condition ensures that the coding framework satisfies the security constraint in \eqref{model:security}.
\end{itemize}
In particular, when all inputs $W_1,W_2,\ldots,W_N$ are independently and uniformly distributed over the finite field $\mathbb{F}_q$, the above two conditions on the encoding matrices are \emph{necessary} for the achievability of the coding framework. 
\end{Proposition}
\begin{IEEEproof}
The condition in \eqref{eq:zerosum-constraint} ensures that the proposed coding framework satisfies the correctness constraint in \eqref{model:correctness}, while the condition in \eqref{eq:rankcondition} further guarantees that the security constraint in \eqref{model:security} is also satisfied. Therefore, the conditions in \eqref{eq:zerosum-constraint} and \eqref{eq:rankcondition} are sufficient for the achievability of the proposed capacity-achieving coding framework.
Furthermore, when all inputs $W_1,W_2,\ldots,W_N$ are independently and uniformly distributed, the conditions in \eqref{eq:zerosum-constraint} and \eqref{eq:rankcondition} are jointly equivalent to the correctness constraint in \eqref{model:correctness} and the security constraint in \eqref{model:security}, and are thus also necessary for the achievability of the coding framework.
The detailed proof is deferred to the appendix.
\end{IEEEproof}


\subsection{Explicit Construction of Encoding Matrices for the Coding Framework}
According to Proposition \ref{pro:condition}, to design a secure aggregation scheme that achieves the optimal rate tuple, it suffices to construct the encoding matrices 
$\{\mathbf{E}_{m,n}\}_{m\in[N],n\in[N]}$ and $\{\mathbf{P}_{m,n},\mathbf{G}_{m,n}\}_{m\in[N]\backslash\{n\},n\in[N]}$ that satisfy the conditions \eqref{eq:structure:cons}--\eqref{eq:rankcondition}. 
In this subsection, we present an explicit construction of these encoding matrices 
over any finite field $\mathbb{F}_q$ with size $q\geq N$, and show that the proposed construction satisfies all the required constraints and gives rise to a capacity-achieving secure aggregation scheme that can be implemented using a pairwise key-distribution structure.



Before presenting the general construction of the encoding matrices, we illustrate the underlying ideas behind the construction through one concrete example.
\begin{Example}\label{example:1}
\label{ex:uncoded}
Consider the secure aggregation problem with system parameters $N=5,T=2$, and $L=3$, where the server aggregates the local inputs of the $N=5$ users, each of length $L=3$, while keeping the input of each individual user private from any collusion between the server and up to $T=2$ other users. We wish to construct the encoding matrices 
$\{\mathbf{E}_{m,n}\in\mathbb{F}_q^{3\times4}\}_{m\in[5],n\in[5]}$ and $\{\mathbf{P}_{m,n}\in\mathbb{F}_q^{3\times1},\mathbf{G}_{m,n}\in\mathbb{F}_q^{1\times 4}\}_{m\in[5]\backslash\{n\},n\in[5]}$ that satisfy the conditions in \eqref{eq:structure:cons}--\eqref{eq:rankcondition} over any finite field $\mathbb{F}_q$ with size $q\geq5$. 


For each $n\in[5]$, user $n$ locally generates $4$ linearly independent row vectors $\{\mathbf{g}_{n,m}\}_{m\in[5]\backslash\{n\}}$ of length $4$ over $\mathbb{F}_q$, and then constructs the encoding matrices $\{\mathbf{G}_{n,m}\}_{m\in[5]\backslash\{n\}}$ as follows:
\begin{IEEEeqnarray*}{rCl}
\text{User 1:} \quad && \mathbf{G}_{1,2} = \mathbf{g}_{1,2}, \ \mathbf{G}_{1,3} = \mathbf{g}_{1,3}, \ \mathbf{G}_{1,4} = \mathbf{g}_{1,4}, \ \mathbf{G}_{1,5} = \mathbf{g}_{1,5}; \quad \ \  \\
\text{User 2:} \quad && \mathbf{G}_{2,1} = \mathbf{g}_{2,1}, \ \mathbf{G}_{2,3} = \mathbf{g}_{2,3}, \ \mathbf{G}_{2,4} = \mathbf{g}_{2,4}, \ \mathbf{G}_{2,5} = \mathbf{g}_{2,5}; \quad \ \  \\
\text{User 3:} \quad && \mathbf{G}_{3,1} = \mathbf{g}_{3,1}, \ \mathbf{G}_{3,2} = \mathbf{g}_{3,2}, \ \mathbf{G}_{3,4} = \mathbf{g}_{3,4}, \ \mathbf{G}_{3,5} = \mathbf{g}_{3,5}; \quad \ \  \\
\text{User 4:} \quad && \mathbf{G}_{4,1} = \mathbf{g}_{4,1}, \ \mathbf{G}_{4,2} = \mathbf{g}_{4,2}, \ \mathbf{G}_{4,3} = \mathbf{g}_{4,3}, \ \mathbf{G}_{4,5} = \mathbf{g}_{4,5}; \quad \ \  \\
\text{User 5:} \quad && \mathbf{G}_{5,1} = \mathbf{g}_{5,1}, \ \mathbf{G}_{5,2} = \mathbf{g}_{5,2}, \ \mathbf{G}_{5,3} = \mathbf{g}_{5,3}, \ \mathbf{G}_{5,4} = \mathbf{g}_{5,4}. 
\end{IEEEeqnarray*}

Let $\alpha_1,\alpha_2,\alpha_3,\alpha_4,\alpha_5$ be $5$ pairwise distinct elements in the finite field $\mathbb{F}_q$. Then, for each $n\in[5]$ and $m\in[5]\backslash\{n\}$, user $n$ locally constructs the encoding matrix $\mathbf{P}_{m,n}$ as
\begin{IEEEeqnarray*}{c}
\mathbf{P}_{m,n} = \begin{bmatrix} 1 \\ \alpha_m \\ \alpha_m^2  \end{bmatrix},\quad \forall\,n\in[5], m\in[5]\backslash\{n\}.
\end{IEEEeqnarray*} 
Accordingly, the encoding matrix
$\mathbf{E}_{m,n}$ is computed as
\begin{IEEEeqnarray*}{c}
\mathbf{E}_{m,n}=\mathbf{P}_{m,n}\mathbf{G}_{m,n} = \begin{bmatrix} \mathbf{g}_{m,n} \\ \alpha_m\mathbf{g}_{m,n} \\ \alpha_m^2\mathbf{g}_{m,n} \end{bmatrix},\quad \forall\,n\in[5], m\in[5]\backslash\{n\}.
\end{IEEEeqnarray*} 
For any $n\in[5]$, user $n$ can locally construct the encoding matrix $\mathbf{E}_{n,n}$ as
\begin{IEEEeqnarray*}{c}
 \mathbf{E}_{n,n}=
-\begin{bmatrix}
    \sum_{m\in[5]\backslash\{n\}}\mathbf{g}_{n,m}\\
      \alpha_n\sum_{m\in[5]\backslash\{n\}}\mathbf{g}_{n,m}\\
    \alpha_n^2\sum_{m\in[5]\backslash\{n\}}\mathbf{g}_{n,m}\\
\end{bmatrix}
\end{IEEEeqnarray*}


We have completed the construction of the encoding matrices over any finite field $\mathbb{F}_q$ with size $q\geq5$, where the global encoding matrix $\mathbf{E}$ of dimensions $15\times 20$ is presented in \eqref{example:1:constructencodematrix}. One can easily verify that the constructed encoding matrices satisfy the conditions in \eqref{eq:structure:cons} and \eqref{eq:zerosum-constraint}.

\begin{figure*}[htbp]
\hrulefill
\begin{IEEEeqnarray}{c}
\renewcommand{\arraystretch}{0.7} 
\mathbf{E}
\!=\!\left[
\setlength{\dashlinedash}{1.5pt}
\setlength{\dashlinegap}{1pt}
\begin{array}{@{}c@{\;} : c@{\;} : c@{\;} : c@{\;} : c@{}}
- \!\!\!\sum\limits_{m\in[5]\setminus\{1\}}\!\!\! \mathbf{g}_{1,m} 
& \mathbf{g}_{2,1} & \mathbf{g}_{3,1} & \mathbf{g}_{4,1} & \mathbf{g}_{5,1} \\
-\alpha_1 \!\!\!\sum\limits_{m\in[5]\setminus\{1\}}\!\!\! \mathbf{g}_{1,m} 
& \alpha_2 \mathbf{g}_{2,1} & \alpha_3 \mathbf{g}_{3,1} & \alpha_4 \mathbf{g}_{4,1} & \alpha_5 \mathbf{g}_{5,1} \\
-\alpha_1^2 \!\!\!\sum\limits_{m\in[5]\setminus\{1\}}\!\!\! \mathbf{g}_{1,m} 
& \alpha_2^2 \mathbf{g}_{2,1} & \alpha_3^2 \mathbf{g}_{3,1} & \alpha_4^2 \mathbf{g}_{4,1} & \alpha_5^2 \mathbf{g}_{5,1} \\
\hdashline

\mathbf{g}_{1,2} 
& - \!\!\!\sum\limits_{m\in[5]\setminus\{2\}}\!\!\! \mathbf{g}_{2,m} 
& \mathbf{g}_{3,2} & \mathbf{g}_{4,2} & \mathbf{g}_{5,2} \\
\alpha_1 \mathbf{g}_{1,2} 
& -\alpha_2 \!\!\!\sum\limits_{m\in[5]\setminus\{2\}}\!\!\! \mathbf{g}_{2,m} 
& \alpha_3 \mathbf{g}_{3,2} & \alpha_4 \mathbf{g}_{4,2} & \alpha_5 \mathbf{g}_{5,2} \\
\alpha_1^2 \mathbf{g}_{1,2} 
& -\alpha_2^2 \!\!\!\sum\limits_{m\in[5]\setminus\{2\}}\!\!\! \mathbf{g}_{2,m} 
& \alpha_3^2 \mathbf{g}_{3,2} & \alpha_4^2 \mathbf{g}_{4,2} & \alpha_5^2 \mathbf{g}_{5,2} \\
\hdashline

\mathbf{g}_{1,3} & \mathbf{g}_{2,3} 
& - \!\!\!\sum\limits_{m\in[5]\setminus\{3\}}\!\!\! \mathbf{g}_{3,m} 
& \mathbf{g}_{4,3} & \mathbf{g}_{5,3} \\
\alpha_1 \mathbf{g}_{1,3} & \alpha_2 \mathbf{g}_{2,3} 
& -\alpha_3 \!\!\!\sum\limits_{m\in[5]\setminus\{3\}}\!\!\! \mathbf{g}_{3,m} 
& \alpha_4 \mathbf{g}_{4,3} & \alpha_5 \mathbf{g}_{5,3} \\
\alpha_1^2 \mathbf{g}_{1,3} & \alpha_2^2 \mathbf{g}_{2,3} 
& -\alpha_3^2 \!\!\!\sum\limits_{m\in[5]\setminus\{3\}}\!\!\! \mathbf{g}_{3,m} 
& \alpha_4^2 \mathbf{g}_{4,3} & \alpha_5^2 \mathbf{g}_{5,3} \\
\hdashline

\mathbf{g}_{1,4} & \mathbf{g}_{2,4} & \mathbf{g}_{3,4} 
& - \!\!\!\sum\limits_{m\in[5]\setminus\{4\}}\!\!\! \mathbf{g}_{4,m} 
& \mathbf{g}_{5,4} \\
\alpha_1 \mathbf{g}_{1,4} & \alpha_2 \mathbf{g}_{2,4} & \alpha_3 \mathbf{g}_{3,4} 
& -\alpha_4 \!\!\!\sum\limits_{m\in[5]\setminus\{4\}}\!\!\! \mathbf{g}_{4,m} 
& \alpha_5 \mathbf{g}_{5,4} \\
\alpha_1^2 \mathbf{g}_{1,4} & \alpha_2^2 \mathbf{g}_{2,4} & \alpha_3^2 \mathbf{g}_{3,4} 
& -\alpha_4^2 \!\!\!\sum\limits_{m\in[5]\setminus\{4\}}\!\!\! \mathbf{g}_{4,m} 
& \alpha_5^2 \mathbf{g}_{5,4} \\
\hdashline

\mathbf{g}_{1,5} & \mathbf{g}_{2,5} & \mathbf{g}_{3,5} & \mathbf{g}_{4,5} 
& - \!\!\!\sum\limits_{m\in[5]\setminus\{5\}}\!\!\! \mathbf{g}_{5,m} \\
\alpha_1 \mathbf{g}_{1,5} & \alpha_2 \mathbf{g}_{2,5} & \alpha_3 \mathbf{g}_{3,5} & \alpha_4 \mathbf{g}_{4,5} 
& -\alpha_5 \!\!\!\sum\limits_{m\in[5]\setminus\{5\}}\!\!\! \mathbf{g}_{5,m} \\
\alpha_1^2 \mathbf{g}_{1,5} & \alpha_2^2 \mathbf{g}_{2,5} & \alpha_3^2 \mathbf{g}_{3,5} & \alpha_4^2 \mathbf{g}_{4,5} 
& -\alpha_5^2 \!\!\!\sum\limits_{m\in[5]\setminus\{5\}}\!\!\! \mathbf{g}_{5,m} \\
\end{array}
\right]. \IEEEeqnarraynumspace
\label{example:1:constructencodematrix} 
\end{IEEEeqnarray}
\hrulefill
\end{figure*}

\begin{figure*}[htbp]
\renewcommand{\arraystretch}{0.70}
\begin{IEEEeqnarray}{rCl}
\begin{bmatrix}
\mathbf{E}_{\{2,3\},\{1,2,3\}} \\ \mathrm{diag}(\mathbf{G}_{\{1,2,3\},\{4,5\}}) \\ \end{bmatrix}
\!&=&\!\left[
\setlength{\dashlinedash}{1.5pt}
\setlength{\dashlinegap}{1pt}
\begin{array}{@{}c@{}:@{}c@{}:@{}c@{}}
\mathbf{g}_{1,2} 
& -\!\sum_{m\in[5]\setminus\{2\}}\!\mathbf{g}_{2,m} 
& \mathbf{g}_{3,2}  \\
\alpha_1\mathbf{g}_{1,2} 
& -\alpha_2\!\sum_{m\in[5]\setminus\{2\}}\!\mathbf{g}_{2,m} 
& \alpha_3\mathbf{g}_{3,2}  \\
\alpha_1^2\mathbf{g}_{1,2} 
& -\alpha_2^2\!\sum_{m\in[5]\setminus\{2\}}\!\mathbf{g}_{2,m} 
& \alpha_3^2\mathbf{g}_{3,2} \\ 
\hdashline

\mathbf{g}_{1,3} 
& \mathbf{g}_{2,3} 
& -\!\sum_{m\in[5]\setminus\{3\}}\!\mathbf{g}_{3,m}  \\
\alpha_1\mathbf{g}_{1,3} 
& \alpha_2\mathbf{g}_{2,3} 
& -\alpha_3\!\sum_{m\in[5]\setminus\{3\}}\!\mathbf{g}_{3,m} \\
\alpha_1^2\mathbf{g}_{1,3} 
& \alpha_2^2\mathbf{g}_{2,3} 
& -\alpha_3^2\!\sum_{m\in[5]\setminus\{3\}}\!\mathbf{g}_{3,m}  \\
\noalign{\hrule height 1pt}

\multicolumn{1}{@{}c@{}:}{\rule{0pt}{2.6ex}\mathbf{g}_{1,4}} &
\multicolumn{1}{c}{}& \\ 

\cdashline{1-2}
& \multicolumn{1}{c:}{\mathbf{g}_{2,4}} & \\  

\cdashline{2-3} 
\multicolumn{1}{c}{}& & \mathbf{g}_{3,4} \\ 
\cdashline{3-3}[1.5pt/1pt]

\cdashline{1-1}
\multicolumn{1}{c:}{\mathbf{g}_{1,5}} &
\multicolumn{1}{c}{} & \\

\cdashline{1-2}[1.5pt/1pt]
& \multicolumn{1}{c:}{\mathbf{g}_{2,5}} & \\

\cdashline{2-3}[1.5pt/1pt]
\multicolumn{1}{c}{}& & \mathbf{g}_{3,5} \\
\end{array}
\right]
\label{eq:matrixexmaple1} \\ 
\!\!&=&\!\!
\arraycolsep=0.9pt
\begin{bNiceMatrix}[columns-width=0.2mm]
\Block[borders={bottom,right,tikz={dash pattern=on 1.5pt off 1pt}}]{3-3}{}
1 & 1 & 1 & & & & & & & & & \\
\alpha_1 & \alpha_2 & \alpha_3 & & & & & & & & & \\
\alpha_1^2 & \alpha_2^2 & \alpha_3^2 & & & & & & & & & \\

& & &
\Block[borders={top,bottom,left,right,tikz={dash pattern=on 1.5pt off 1pt}}]{3-3}{}
1 & 1 & 1 & & & & & & \\
& & &
\alpha_1 & \alpha_2 & \alpha_3 & & & & & & \\
& & &
\alpha_1^2 & \alpha_2^2 & \alpha_3^2 & & & & & & \\

& & & & & &
\Block[borders={top,bottom,left,right,tikz={dash pattern=on 1.5pt off 1pt}}]{3-3}{}
\hbox to 1em{\hfil$1$\hfil} &
\hbox to 1em{\hfil$0$\hfil} &
\hbox to 1em{\hfil$0$\hfil} & & & \\

& & & & & &
\hbox to 1em{\hfil$0$\hfil} &
\hbox to 1em{\hfil$1$\hfil} &
\hbox to 1em{\hfil$0$\hfil} & & & \\

& & & & & &
\hbox to 1em{\hfil$0$\hfil} &
\hbox to 1em{\hfil$0$\hfil} &
\hbox to 1em{\hfil$1$\hfil} & & & \\

& & & & & & & & &
\Block[borders={top,left,tikz={dash pattern=on 1.5pt off 1pt}}]{3-3}{}
\hbox to 1em{\hfil$1$\hfil} &
\hbox to 1em{\hfil$0$\hfil} &
\hbox to 1em{\hfil$0$\hfil} \\

& & & & & & & & &
\hbox to 1em{\hfil$0$\hfil} &
\hbox to 1em{\hfil$1$\hfil} &
\hbox to 1em{\hfil$0$\hfil} \\

& & & & & & & & &
\hbox to 1em{\hfil$0$\hfil} &
\hbox to 1em{\hfil$0$\hfil} &
\hbox to 1em{\hfil$1$\hfil}
\end{bNiceMatrix}
\setlength{\arraycolsep}{0.2pt}
\begin{bNiceMatrix}[columns-width=8mm]
\Block[borders={bottom,right,tikz={dash pattern=on 1.5pt off 1pt}}]{1-1}{\mathbf{g}_{1,2}}
& & \\

&
\Block[borders={top,bottom,left,right,tikz={dash pattern=on 1.5pt off 1pt}}]{1-1}{
-\!\sum_{m\in[5]\setminus\{2\}}\!\mathbf{g}_{2,m}}
& \\

& &
\Block[borders={top,bottom,left,tikz={dash pattern=on 1.5pt off 1pt}}]{1-1}{\mathbf{g}_{3,2}} \\

\Block[borders={top,bottom,right,tikz={dash pattern=on 1.5pt off 1pt}}]{1-1}{\mathbf{g}_{1,3}}
& & \\

&
\Block[borders={top,bottom,left,right,tikz={dash pattern=on 1.5pt off 1pt}}]{1-1}{\mathbf{g}_{2,3}}
& \\

& &
\Block[borders={top,bottom,left,tikz={dash pattern=on 1.5pt off 1pt}}]{1-1}{-\!\sum_{m\in[5]\setminus\{3\}}\!\mathbf{g}_{3,m}} \\

\Block[borders={top,bottom,right,tikz={dash pattern=on 1.5pt off 1pt}}]{1-1}{\mathbf{g}_{1,4}}
& & \\

&
\Block[borders={top,bottom,left,right,tikz={dash pattern=on 1.5pt off 1pt}}]{1-1}{\mathbf{g}_{2,4}}
& \\

& &
\Block[borders={top,bottom,left,tikz={dash pattern=on 1.5pt off 1pt}}]{1-1}{\mathbf{g}_{3,4}} \\

\Block[borders={top,bottom,right,tikz={dash pattern=on 1.5pt off 1pt}}]{1-1}{\mathbf{g}_{1,5}}
& & \\

&
\Block[borders={top,bottom,left,right,tikz={dash pattern=on 1.5pt off 1pt}}]{1-1}{\mathbf{g}_{2,5}}
& \\

& &
\Block[borders={top,left,tikz={dash pattern=on 1.5pt off 1pt}}]{1-1}{\mathbf{g}_{3,5}}
\end{bNiceMatrix}.
\IEEEeqnarraynumspace \label{eq:matrixexmaple2}
\end{IEEEeqnarray}
\hrulefill
\end{figure*}


Next, we verify that the condition in \eqref{eq:rankcondition} is satisfied. 
Without loss of generality, assume that users $4$ and $5$ collude with the server.
Let $\mathcal{T}\!=\!\{4,5\}$ and $\mathcal{T}^c\!=\!\{1,2,3\}$, and choose $\widetilde{\mathcal{T}}^c\!=\!\{2,3\}$. Then, the matrix $\mathbf{E}_{\{2,3\},\{1,2,3\}}$ of dimensions $6\times 12$ and the block-diagonal matrix $\mathrm{diag}(\mathbf{G}_{\{1,2,3\},\{4,5\}})$ of dimensions $6\times 12$ have the form given in \eqref{eq:matrixexmaple1}. Since the row vectors $\{\mathbf{g}_{n,m}\}_{m\in[5]\backslash\{n\}}$ are linearly independent for each $n\in[5]$ and the elements $\{\alpha_1,\alpha_2,\alpha_3,\alpha_4,\alpha_5\}$ are pairwise distinct over $\mathbb{F}_q$, it can be checked that the stacked matrix
$\Bigg[\begin{array}{@{\,}c@{\,}} \mathbf{E}_{\{2,3\},\{1,2,3\}} \\ \mathrm{diag}(\mathbf{G}_{\{1,2,3\},\{4,5\}}) \\ \end{array}\Bigg]$
of dimensions $12\times 12$ has full rank by \eqref{eq:matrixexmaple2}.
Moreover, the  submatrix $\mathrm{diag}(\mathbf{G}_{\{1,2,3\},\{4,5\}})$ of the stacked matrix has full row rank. Therefore, we have
\begin{IEEEeqnarray*}{c}
\operatorname{rank}\left(\begin{bmatrix}
\mathbf{E}_{\{2,3\},\{1,2,3\}} \\ \mathrm{diag}(\mathbf{G}_{\{1,2,3\},\{4,5\}}) \\ \end{bmatrix} \right)-
\operatorname{rank}\left(\mathrm{diag}(\mathbf{G}_{\{1,2,3\},\{4,5\}}) \right)=12-6=6,
\end{IEEEeqnarray*}
which satisfies the condition in \eqref{eq:rankcondition}. The remaining colluding cases can be established using similar arguments. Consequently, the constructed encoding matrices satisfy the conditions in \eqref{eq:structure:cons}--\eqref{eq:rankcondition} over any finite field $\mathbb{F}_q$ with size $q\geq5$.
\hfill \IEEEQED
\end{Example}

In the following, we present the general construction of the encoding matrices 
that satisfy the conditions \eqref{eq:structure:cons}--\eqref{eq:rankcondition} over any finite field $\mathbb{F}_q$ with size $q\geq N$.

Let $\mathbb{F}_q$ be an arbitrary finite field with size $q\geq N$.
For each $n\in[N]$, user $n$ locally generates a full-rank square matrix $\mathbf{G}_{n}$ of dimensions $(N-1)\times(N-1)$ over $\mathbb{F}_q$, e.g., the identity matrix $\mathbf{I}_{N-1}$, given by
\begin{IEEEeqnarray}{c}\label{definition:Gn}
\mathbf{G}_{n}=
\begin{bmatrix}
\mathbf{g}_{n,1}\\
\vdots\\
\mathbf{g}_{n,n-1}\\
\mathbf{g}_{n,n+1}\\
\vdots\\
\mathbf{g}_{n,N}
\end{bmatrix}
, \quad \forall\,n\in[N],
\end{IEEEeqnarray}
where $\{\mathbf{g}_{n,m}\}_{m\in[N]\backslash\{n\}}$ denote the $N-1$ row vectors of the matrix $\mathbf{G}_{n}$, each of length $N-1$, which are linearly independent over $\mathbb{F}_q$. Then, the encoding matrix $\mathbf{G}_{n,m}$ of dimensions $\frac{L}{N-T}\times\frac{N-1}{N-T}L$ is constructed as
\begin{IEEEeqnarray}{c}\label{definition:Gnm}
\mathbf{G}_{n,m}=\mathbf{g}_{n,m}\otimes\mathbf{I}_{\frac{L}{N-T}} ,        \quad\forall\, n\in[N], m\in[N]\backslash\{n\}.
\end{IEEEeqnarray}


Next, we describe the construction of the encoding matrices $\{\mathbf{P}_{m,n}\}_{m\in[N]\backslash\{n\},n\in[N]}$. Let $\alpha_1,\alpha_2,\ldots,\alpha_N$ be $N$ pairwise distinct elements in $\mathbb{F}_q$, which are publicly known to all users. Then, 
user $n$ can locally generate the encoding matrix $\mathbf{P}_{m,n}$ of dimensions $L\times\frac{L}{N-T}$ as
\begin{IEEEeqnarray}{c}
\mathbf{P}_{m,n}
=\begin{bmatrix} 
1\\
\alpha_m\\
\vdots\\
\alpha_m^{N-T-1}
\end{bmatrix}\otimes\mathbf{I}_{\frac{L}{N-T}}
, \quad\forall\, n\in[N], m\in[N]\backslash\{n\}. 
\end{IEEEeqnarray}
Accordingly, the encoding matrix $\mathbf{E}_{m,n}$ of dimensions $L\times \frac{N-1}{N-T}L$ is given by 
\begin{IEEEeqnarray}{c}\label{construction:E:mn}
\mathbf{E}_{m,n}=\mathbf{P}_{m,n}\mathbf{G}_{m,n}
=\begin{bmatrix}
\mathbf{g}_{m,n}\\
\alpha_m\mathbf{g}_{m,n}\\
\vdots\\
\alpha_m^{N-T-1}\mathbf{g}_{m,n}
\end{bmatrix}\otimes\mathbf{I}_{\frac{L}{N-T}}
, \quad\forall\, n\in[N], m\in[N]\backslash\{n\}.
\end{IEEEeqnarray}
To ensure that the constructed encoding matrices satisfy the prescribed condition in \eqref{eq:zerosum-constraint}, we design the encoding matrix $\mathbf{E}_{n,n}$ of dimensions $L\times \frac{N-1}{N-T}L$ as 
\begin{IEEEeqnarray}{c}\label{construction:E:nn}
\mathbf{E}_{n,n}=-\sum\limits_{m\in[N]\backslash\{n\}}\mathbf{E}_{n,m}
=-\begin{bmatrix} 
\sum_{m\in[N]\backslash\{n\}}\mathbf{g}_{n,m}\\
\alpha_n\sum_{m\in[N]\backslash\{n\}}\mathbf{g}_{n,m}\\
\vdots\\
\alpha_n^{N-T-1}\sum_{m\in[N]\backslash\{n\}}\mathbf{g}_{n,m}
\end{bmatrix}\otimes\mathbf{I}_{\frac{L}{N-T}}
, \quad\forall\,n\in[N].
\end{IEEEeqnarray}  
which is a deterministic function of the matrices $\{\mathbf{G}_{n,m}\}_{n\in[N],m\in[N]\backslash\{n\}}$ available at user $n$ and therefore can be locally generated by user $n$.
This implies that although our secure aggregation scheme involves a large number of encoding matrices $\{\mathbf{E}_{m,n}\}_{m\in[N],n\in[N]}$ and $\{\mathbf{P}_{m,n},\mathbf{G}_{m,n}\}_{m\in[N]\backslash\{n\},n\in[N]}$, each user only needs to store the $N$ encoding parameters $\{\alpha_n\}_{n\in[N]}$ in practice. All the required encoding matrices can be generated locally from these parameters, incurring no additional storage overhead.

Having completed the construction of $\{\mathbf{E}_{m,n}\}_{m\in[N],n\in[N]}$, the global encoding matrix $\mathbf{E}$ of dimensions $NL\times \frac{N(N-1)}{N-T}L$ can be written as
\begin{IEEEeqnarray}{c}
\mathbf{E}
\setlength{\dashlinedash}{1.5pt}
\setlength{\dashlinegap}{1pt}
=\left[
\begin{array}{@{}c:c:c:c@{}}
\begin{matrix} 
\alpha_1^{0}\mathbf{g}_{1,1} \\ \vdots \\ \alpha_1^{N-T-1}\mathbf{g}_{1,1}
\end{matrix} & 
\begin{matrix} 
\alpha_2^0\mathbf{g}_{2,1} \\ \vdots \\ \alpha_2^{N-T-1}\mathbf{g}_{2,1} 
\end{matrix} & 
\cdots & 
\begin{matrix} 
\alpha_N^0\mathbf{g}_{N,1} \\ \vdots \\ \alpha_N^{N-T-1}\mathbf{g}_{N,1}  
\end{matrix} \\ \hdashline
\begin{matrix} 
\alpha_1^0\mathbf{g}_{1,2} \\ \vdots \\ \alpha_1^{N-T-1}\mathbf{g}_{1,2} 
\end{matrix} & 
\begin{matrix} 
\alpha_2^{0}\mathbf{g}_{2,2} \\ \vdots \\ \alpha_2^{N-T-1}\mathbf{g}_{2,2}
\end{matrix} & 
\cdots & 
\begin{matrix} 
\alpha_N^0\mathbf{g}_{N,2} \\ \vdots \\ \alpha_N^{N-T-1}\mathbf{g}_{N,2} 
\end{matrix} \\ \hdashline
\vdots & \vdots & \cdots & \vdots \\ \hdashline
\begin{matrix} 
\alpha_1^0\mathbf{g}_{1,N} \\ \vdots \\ \alpha_1^{N-T-1}\mathbf{g}_{1,N}  
\end{matrix} & 
\begin{matrix} 
\alpha_2^0\mathbf{g}_{2,N} \\ \vdots \\ \alpha_2^{N-T-1}\mathbf{g}_{2,N} 
\end{matrix} & 
\cdots & 
\begin{matrix} 
\alpha_N^{0}\mathbf{g}_{N,N} \\ \vdots \\ \alpha_N^{N-T-1}\mathbf{g}_{N,N}
\end{matrix}
\end{array}
\right]\otimes\mathbf{I}_{\frac{L}{N-T}},\nonumber \IEEEeqnarraynumspace
\end{IEEEeqnarray}
where for simplicity, the row vector $\mathbf{g}_{n,n}$ of length $N-1$ is defined as
\begin{IEEEeqnarray}{c}\label{definition:gnn}
\mathbf{g}_{n,n}=-\sum\limits_{m\in[N]\backslash\{n\}}\mathbf{g}_{n,m},\quad\forall\, n\in[N].
\end{IEEEeqnarray}

In general, the above construction can be viewed as first designing encoding matrices for the base case $L=N-T$, and then extending the construction to arbitrary input length $L$ via the Kronecker product with the identity matrix $\mathbf{I}_{\frac{L}{N-T}}$ of dimensions $\frac{L}{N-T}\times \frac{L}{N-T}$. 
The verification that the constructed encoding matrices satisfy the conditions in \eqref{eq:structure:cons}--\eqref{eq:rankcondition} follows the same lifting principle: it suffices to establish the required conditions for the base construction with $L=N-T$, after which the general case follows directly from the Kronecker-product structure.
\begin{Lemma}\label{lemma:encoding:matrix}
For any system parameters $N>1$ and $0\leq T<N-1$, the  encoding matrices $\{\mathbf{E}_{m,n}\}_{m\in[N],n\in[N]}$ and $\{\mathbf{P}_{m,n},\mathbf{G}_{m,n}\}_{m\in[N]\backslash\{n\},n\in[N]}$ constructed in \eqref{definition:Gnm}--\eqref{construction:E:nn} satisfy the conditions in
\eqref{eq:structure:cons}--\eqref{eq:rankcondition} over any finite field $\mathbb{F}_q$ with size $q\geq N$.
\end{Lemma}

\begin{IEEEproof}
From \eqref{construction:E:mn} and \eqref{construction:E:nn}, it is straightforward  to verify that the constructed encoding matrices satisfy the conditions in \eqref{eq:structure:cons} and \eqref{eq:zerosum-constraint}. We then proceed to show that the condition in \eqref{eq:rankcondition} is also satisfied.  


For any subset of colluding users $\mathcal{T}\subseteq[N]$ with size $|\mathcal{T}|\leq T$ and its complement $\mathcal{T}^c=[N]\backslash\mathcal{T}$ with size $|\mathcal{T}^c|=N-|\mathcal{T}|$, we index the users in $\mathcal{T}^c$ as $\{n_1,n_2,\ldots,n_{N-|\mathcal{T}|}\}$ and those in $\mathcal{T}$ as $\{n_{N-|\mathcal{T}|+1},n_{N-|\mathcal{T}|+2},\dots,n_{N}\}$.
Let $\widetilde{\mathcal{T}}^{c}=\{n_2,\dots,n_{N-|\mathcal{T}|}\}$, then the matrix $\mathbf{E}_{\widetilde{\mathcal{T}}^c,\mathcal{T}^c}$ of dimensions $(N-|\mathcal{T}|-1)L\times\frac{(N-|\mathcal{T}|)(N-1)}{N-T}L$ and the matrix $\mathrm{diag}(\mathbf{G}_{\mathcal{T}^c,\mathcal{T}})$ of dimensions $\frac{|\mathcal{T}|(N-|\mathcal{T}|)}{N-T}L\times\frac{(N-|\mathcal{T}|)(N-1)}{N-T}L$ have the following form:
\begin{IEEEeqnarray}{c}
\renewcommand{\arraystretch}{0.7} 
\left[
  \begin{array}{@{}c@{}}
\mathbf{E}_{\widetilde{\mathcal{T}}^c,\mathcal{T}^c} \\
\mathrm{diag}(\mathbf{G}_{\mathcal{T}^c,\mathcal{T}})\\
\end{array}
\right] =
\setlength{\dashlinedash}{1.5pt}
\setlength{\dashlinegap}{1pt}
\left[
\begin{array}{@{}c:c:c:c@{}}
\alpha_{n_1}^0\mathbf{g}_{n_1,n_2} & \alpha_{n_2}^{0}\mathbf{g}_{n_2,n_2} & \cdots & \alpha_{n_{N-|\mathcal{T}|}}^0\mathbf{g}_{n_{N-|\mathcal{T}|},n_2} \\
\vdots & \vdots & \vdots & \vdots \\
\alpha_{n_1}^{N-T-1}\mathbf{g}_{n_1,n_2} & \alpha_{n_2}^{N-T-1}\mathbf{g}_{n_2,n_2} & \cdots & \alpha_{n_{N-|\mathcal{T}|}}^{N-T-1}\mathbf{g}_{n_{N-|\mathcal{T}|},n_2} \\
\hdashline
\alpha_{n_1}^0\mathbf{g}_{n_1,n_3} & \alpha_{n_2}^{0}\mathbf{g}_{n_2,n_3} & \cdots & \alpha_{n_{N-|\mathcal{T}|}}^0\mathbf{g}_{n_{N-|\mathcal{T}|},n_3} \\
\vdots & \vdots & \vdots & \vdots \\
\alpha_{n_1}^{N-T-1}\mathbf{g}_{n_1,n_3} & \alpha_{n_2}^{N-T-1}\mathbf{g}_{n_2,n_3} & \cdots & \alpha_{n_{N-|\mathcal{T}|}}^{N-T-1}\mathbf{g}_{n_{N-|\mathcal{T}|},n_3} \\
\hdashline
\vdots & \vdots & \vdots & \vdots \\
\hdashline
\alpha_{n_1}^0\mathbf{g}_{n_1,n_{N-|\mathcal{T}|}} & \alpha_{n_2}^0\mathbf{g}_{n_2,n_{N-|\mathcal{T}|}} & \cdots & \alpha_{n_{N-|\mathcal{T}|}}^{0} \mathbf{g}_{n_{N-|\mathcal{T}|},n_{N-|\mathcal{T}|}} \\
\vdots & \vdots & \vdots & \vdots \\
\alpha_{n_1}^{N-T-1}\mathbf{g}_{n_1,n_{N-|\mathcal{T}|}} & \alpha_{n_2}^{N-T-1}\mathbf{g}_{n_2,n_{N-|\mathcal{T}|}} & \cdots & \alpha_{n_{N-|\mathcal{T}|}}^{N-T-1}\mathbf{g}_{n_{N-|\mathcal{T}|},n_{N-|\mathcal{T}|}} \rule[-2ex]{0pt}{0pt}\\
\noalign{\hrule height 1.0pt}

\multicolumn{1}{c}{\rule{0pt}{2.6ex}\mathbf{g}_{n_1,n_{N-|\mathcal{T}|+1}}} & \multicolumn{1}{c}{} & \multicolumn{1}{c}{}&  \\
 \multicolumn{1}{c}{}& \multicolumn{1}{c}{\mathbf{g}_{n_2,n_{N-|\mathcal{T}|+1}}} & \multicolumn{1}{c}{}&  \\
\multicolumn{1}{c}{}& \multicolumn{1}{c}{} & \multicolumn{1}{c}{\ddots} & \\
\multicolumn{1}{c}{}& \multicolumn{1}{c}{}&  \multicolumn{1}{c}{}& \mathbf{g}_{n_{N-|\mathcal{T}|},n_{N-|\mathcal{T}|+1}} \\
\hdashline
\multicolumn{1}{c}{\vdots} & \multicolumn{1}{c}{\vdots }&\multicolumn{1}{c}{\cdots } & \vdots \\
\hdashline

 \multicolumn{1}{c}{\mathbf{g}_{n_1,n_{N}}} & \multicolumn{1}{c}{} & \multicolumn{1}{c}{}&   \\
 \multicolumn{1}{c}{}& \multicolumn{1}{c}{\mathbf{g}_{n_2,n_{N}}} & \multicolumn{1}{c}{} &  \\
\multicolumn{1}{c}{} & \multicolumn{1}{c}{} & \multicolumn{1}{c}{\ddots}& \\
\multicolumn{1}{c}{} & \multicolumn{1}{c}{} & \multicolumn{1}{c}{} & \mathbf{g}_{n_{N-|\mathcal{T}|},n_{N}} \\
\end{array}
\right]
\otimes\mathbf{I}_{\frac{L}{N-T}}
, 
\label{general:cons:Rank:222} 
\IEEEeqnarraynumspace
\end{IEEEeqnarray}
which can be equivalently written as
\begin{IEEEeqnarray}{c}
\left[
  \begin{array}{@{}c@{}}
\mathbf{E}_{\widetilde{\mathcal{T}}^c,\mathcal{T}^c} \\
\mathrm{diag}(\mathbf{G}_{\mathcal{T}^c,\mathcal{T}})\\
\end{array}
\right]=\!\!\!\!\!\!\!\!\!\!
%
\begin{array}{cc}
\begin{array}{r}
    \text{\scriptsize $N\!-\!|\mathcal{T}|\!-\!1~$blocks}   \left\{ \vphantom{\begin{matrix} \mathbf{P} \\ \ddots \\ \mathbf{P} \end{matrix}} \right. \\
    \text{\scriptsize $|\mathcal{T}|$ blocks}    \left\{ \vphantom{\begin{matrix} \mathbf{I} \\ \ddots \\ \mathbf{I} \end{matrix}} \right.
\end{array}
\!\!\!\!\!\!
& 
\hspace{-1em} 
\underbrace{
\left[
\begin{array}{@{}c@{}c@{}c@{}c@{}c@{}c@{}}
\mathbf{P} & & & & &\\
& \ddots & & & &\\
&  & \mathbf{P}& & &\\
&  & & \mathbf{I}_{N-|\mathcal{T}|}& &\\
& &  & & \ddots &\\
& & & & &\mathbf{I}_{N-|\mathcal{T}|} \\
\end{array}
\right]
}_{\triangleq \mathbf{D}_1}
\end{array}
\!\!\!\!
\underbrace{\begin{bmatrix}
\mathrm{diag}(\mathbf{g}_{n_1,n_2}, \mathbf{g}_{n_2,n_2}, \cdots,\mathbf{g}_{n_{N-|\mathcal{T}|},n_2})\\
\mathrm{diag}(\mathbf{g}_{n_1,n_3},\mathbf{g}_{n_2,n_3},\cdots, \mathbf{g}_{n_{N-|\mathcal{T}|},n_3})\\
\vdots\\
\mathrm{diag}(\mathbf{g}_{n_1,n_{N}}, \mathbf{g}_{n_2,n_{N}}, \cdots,  \mathbf{g}_{n_{N-|\mathcal{T}|},n_{N}} )
\end{bmatrix}}_{\triangleq\mathbf{D}_2}\!\otimes\mathbf{I}_{\frac{L}{N-T}}, \IEEEeqnarraynumspace \label{encoding:proof:222}
\end{IEEEeqnarray} 
where $\mathbf{P}$ is a matrix of dimensions $(N-T)\times(N-|\mathcal{T}|)$, given by
\begin{IEEEeqnarray}{rCl}
\mathbf{P} 
\renewcommand{\arraystretch}{0.5} 
&=&
\left[
\begin{array}{@{}cccc@{}}
1 & 1 & \cdots & 1\\
\alpha_{n_1} & \alpha_{n_2} & \cdots & \alpha_{n_{N-|\mathcal{T}|}}\\
\vdots & \vdots & \cdots & \vdots\\
\alpha_{n_1}^{N-T-1} & \alpha_{n_2}^{N-T-1} & \cdots & \alpha_{n_{N-|\mathcal{T}|}}^{N-T-1}
\end{array}
\right].
\label{encoding:proof:111}
\end{IEEEeqnarray}

Since the elements $\alpha_{n_1},\alpha_{n_2},\ldots,\alpha_{n_{N-|\mathcal{T}|}}$ are pairwise distinct over $\mathbb{F}_q$, the $(N-T)\times (N-|\mathcal{T}|)$ Vandermonde matrix $\mathbf{P}$ defined in \eqref{encoding:proof:111} has full row rank, i.e., $\operatorname{rank}\big(\mathbf{P}\big)=N-T$.
Therefore, the rank of the $\big((N-|\mathcal{T}|-1)(N-T)+|\mathcal{T}|(N-|\mathcal{T}|)\big)\times(N-|\mathcal{T}|)(N-1)$ block-diagonal matrix $\mathbf{D}_1$ defined in \eqref{encoding:proof:222} has
\begin{IEEEeqnarray}{c}\label{general:rank:D1}
\operatorname{rank}\big(\mathbf{D}_1\big)=(N-|\mathcal{T}|-1)\cdot\operatorname{rank}\big(\mathbf{P}\big)+|\mathcal{T}|\cdot\operatorname{rank}\big(\mathbf{I}_{N-|\mathcal{T}|}\big)
=(N-|\mathcal{T}|-1)(N-T)+|\mathcal{T}|(N-|\mathcal{T}|), \IEEEeqnarraynumspace
\end{IEEEeqnarray}
and the rank of the $(N-|\mathcal{T}|)(N-1)\times(N-|\mathcal{T}|)(N-1)$ matrix $\mathbf{D}_2$ defined in \eqref{encoding:proof:222} satisfies
\begin{IEEEeqnarray}{rCl}
\operatorname{rank}\big(
\mathbf{D}_2
\big)
&\overset{(a)}{=}&
\operatorname{rank}\Big(
\left[
\setlength{\dashlinedash}{1.5pt}
\setlength{\dashlinegap}{1pt}
\renewcommand{\arraystretch}{0.7}
\begin{array}{cccc}
    \multicolumn{1}{c:}{\mathbf{g}_{n_1,n_2}}  & & &\\
    \multicolumn{1}{c:}{\mathbf{g}_{n_1,n_3}}& & &\\
   \multicolumn{1}{c:}{\vdots}& & &\\
    \multicolumn{1}{c:}{\mathbf{g}_{n_1,n_N}}& & &\\
    \cdashline{1-2}
    \multicolumn{1}{c:}{}&\multicolumn{1}{c:}{\mathbf{g}_{n_2,n_2}} & &\\
    \multicolumn{1}{c:}{}&\multicolumn{1}{c:}{\mathbf{g}_{n_2,n_3}} & &\\
    \multicolumn{1}{c:}{}&\multicolumn{1}{c:}{\vdots} & &\\
    \multicolumn{1}{c:}{}&\multicolumn{1}{c:}{\mathbf{g}_{n_2,n_{N}}} & &\\
    \cdashline{2-2}
    & & \ddots&\\
    \cdashline{4-4}
    &&\multicolumn{1}{c:}{}&\mathbf{g}_{n_{N-|\mathcal{T}|},n_2}\\
    &&\multicolumn{1}{c:}{}&\mathbf{g}_{n_{N-|\mathcal{T}|},n_3}\\
    &&\multicolumn{1}{c:}{}&\mathbf{g}_{n_{N-|\mathcal{T}|},n_N}\\
\end{array}
\right]
\Big)\notag\\
&\overset{(b)}{=}&\operatorname{rank}\big(\mathrm{diag}(\mathbf{G}_{n_1},\mathbf{G}_{n_2},\cdots,\mathbf{G}_{n_{N-|\mathcal{T}|}})\big) \notag\\
&=&\operatorname{rank}\big(\mathbf{G}_{n_1}\big)+\operatorname{rank}\big(\mathbf{G}_{n_2}\big)+\ldots+\operatorname{rank}\big(\mathbf{G}_{n_{N-|\mathcal{T}|}}\big) \notag\\
&\overset{(c)}{=}&(N-|\mathcal{T}|)(N-1), \label{general:proof:333}
\end{IEEEeqnarray}
where $(a)$ follows from the fact that elementary row operations do not change the rank of a matrix; $(b)$ follows by \eqref{definition:Gn} and \eqref{definition:gnn}, along with the fact that elementary row operations preserve rank; $(c)$ holds because $\mathbf{G}_n$ is a full-rank square matrix of dimensions $(N-1)\times(N-1)$ for any $n\in[N]$ by \eqref{definition:Gn}.

Furthermore, the stacked matrix $\Bigg[
  \begin{array}{@{}c@{}}
\mathbf{E}_{\widetilde{\mathcal{T}}^c,\mathcal{T}^c} \\
\mathrm{diag}(\mathbf{G}_{\mathcal{T}^c,\mathcal{T}})\\
\end{array}
\Bigg]$ of dimensions $(N-|\mathcal{T}|-1+\frac{|\mathcal{T}|(N-|\mathcal{T}|)}{N-T})L\times\frac{(N-|\mathcal{T}|)(N-1)}{N-T}L$ satisfies
\begin{IEEEeqnarray}{rCl}
\operatorname{rank}\big(\left[
  \begin{array}{@{}c@{}}
\mathbf{E}_{\widetilde{\mathcal{T}}^c,\mathcal{T}^c} \\
\mathrm{diag}(\mathbf{G}_{\mathcal{T}^c,\mathcal{T}})\\
\end{array}
\right]\big)&\overset{(a)}{=}&
\operatorname{rank}\big((
\mathbf{D}_1\cdot\mathbf{D}_2) \otimes \mathbf{I}_{\frac{L}{N-T}}\big) \notag\\
&\overset{(b)}{=}&\operatorname{rank}\big(\mathbf{D}_1\cdot \mathbf{D}_2 \big)\times \operatorname{rank}\big(\mathbf{I}_{\frac{L}{N-T}}\big) \notag\\
&\overset{(c)}{=}&\operatorname{rank}\big(\mathbf{D}_1 \big)\times \operatorname{rank}\big(\mathbf{I}_{\frac{L}{N-T}}\big) \notag\\
&\overset{(d)}{=}&\big(N-|\mathcal{T}|-1+\frac{|\mathcal{T}|(N-|\mathcal{T}|)}{N-T}\big)L,
\notag
\end{IEEEeqnarray}
where $(a)$ is due to \eqref{encoding:proof:222};
$(b)$ follows from the fact that the rank of a Kronecker product equals the product of the ranks of the individual matrices \cite[Theorem 4.2.15]{horn1994topics};
$(c)$ follows by the fact that $\mathbf{D}_2$ is a full-rank square matrix by
\eqref{general:proof:333}; $(d)$ is due to \eqref{general:rank:D1}.
That is, the stacked matrix $\Bigg[
  \begin{array}{@{}c@{}}
\mathbf{E}_{\widetilde{\mathcal{T}}^c,\mathcal{T}^c} \\
\mathrm{diag}(\mathbf{G}_{\mathcal{T}^c,\mathcal{T}})\\
\end{array}
\Bigg]$ has full row rank. 

As $\mathrm{diag}(\mathbf{G}_{\mathcal{T}^c,\mathcal{T}})$ is a submatrix of the above stacked matrix with dimensions $\frac{|\mathcal{T}|(N-|\mathcal{T}|)}{N-T}L\times\frac{(N-|\mathcal{T}|)(N-1)}{N-T}L$, it also has full row rank, i.e., 
\begin{IEEEeqnarray}{rCl}
\operatorname{rank}\big(\mathrm{diag}(\mathbf{G}_{\mathcal{T}^c,\mathcal{T}})\big)=\frac{|\mathcal{T}|(N-|\mathcal{T}|)}{N-T}L.
\notag
\end{IEEEeqnarray}

Consequently, we obtain
\begin{IEEEeqnarray}{rCl}
\renewcommand{\arraystretch}{0.5}
\operatorname{rank}\left(
\bigg[
  \begin{array}{@{}c@{}}
\mathbf{E}_{\widetilde{\mathcal{T}}^c,\mathcal{T}^c} \\
\mathrm{diag}(\mathbf{G}_{\mathcal{T}^c,\mathcal{T}})\\
\end{array}
\bigg]
\right)
 - \operatorname{rank}\left(\mathrm{diag}(\mathbf{G}_{\mathcal{T}^c,\mathcal{T}})\right)&=&\Big(N-|\mathcal{T}|-1+\frac{|\mathcal{T}|(N-|\mathcal{T}|)}{N-T}\Big)L-\frac{|\mathcal{T}|(N-|\mathcal{T}|)}{N-T}L\nonumber\\
 &=&(N-|\mathcal{T}|-1)L, \notag
\end{IEEEeqnarray} 
which satisfies the condition in \eqref{eq:rankcondition}. Therefore, the constructed encoding matrices satisfy the conditions in
\eqref{eq:structure:cons}--\eqref{eq:rankcondition} over any finite field $\mathbb{F}_q$ with size $q\geq N$. This completes the proof of the lemma.
\end{IEEEproof}

By applying Lemma \ref{lemma:encoding:matrix} to Proposition \ref{pro:condition}, the coding framework can be implemented using the encoding matrices constructed in this subsection. As a result, we obtain a secure aggregation scheme that achieves the optimal rate tuple over any finite field $\mathbb{F}_q$ with size $q\geq N$. We state this result in the following lemma.
\begin{Lemma}
For the secure aggregation problem with system parameters $N>1$ and $0\leq T<N-1$, we can obtain an explicit secure aggregation scheme using the encoding matrices constructed in \eqref{definition:Gnm}--\eqref{construction:E:nn}, which achieves the optimal rate tuple $\big(\frac{N (N-1)}{N - T},\frac{N (N-1)}{N - T},N\big)$ over any finite field $\mathbb{F}_q$ with size $q\geq N$.
\end{Lemma}

The achievability of Theorems \ref{theo:capacity} and \ref{theo:achive} follows directly from this lemma.

\begin{Lemma}
The proposed capacity-achieving secure aggregation scheme can be achieved using a pairwise key-distribution structure, in which every pair of users shares a mutually independent common random key.
\end{Lemma}
\begin{IEEEproof}
Since each pairwise key $Z_{n,m}$ is shared between users $n$ and $m$ for any $n\in[N]$ and $m\in[N]\backslash\{n\}$, it suffices to prove the lemma by showing that the proposed secure aggregation scheme satisfies the following two conditions: 1) all the encoded keys $\{Z_{n,m}\}_{n\in[N],m\in[N]\backslash\{n\}}$ are mutually independent and uniformly distributed over $\mathbb{F}_q$; and 2) the masked input $X_n$ is generated by using only the pairwise keys $\{Z_{n,m},Z_{m,n}\}_{m\in[N]\backslash\{n\}}$ available at user $n$ to mask the private input $W_n$ for any $n\in[N]$.

According to \eqref{eq:matrixofcodedkey} and \eqref{definition:Gnm}, for any $n\in[N]$ and $m\in[N]\backslash\{n\}$, the pairwise key $Z_{n,m}$ shared by users $n$ and $m$ is given by
\begin{IEEEeqnarray}{c}\notag
Z_{n,m}=\mathbf{G}_{n,m}Z_n=\big(\mathbf{g}_{n,m}\otimes\mathbf{I}_{\frac{L}{N-T}}\big)Z_n.
\end{IEEEeqnarray}
Accordingly, by \eqref{definition:Gn}, we have
\begin{IEEEeqnarray}{c}\label{pairwise:1}
\begin{bmatrix}
Z_{n,1}\\
\vdots\\
Z_{n,n-1}\\
Z_{n,n+1}\\
\vdots\\
Z_{n,N}
\end{bmatrix}=
\Big(\begin{bmatrix}
\mathbf{g}_{n,1}\\
\vdots\\
\mathbf{g}_{n,n-1}\\
\mathbf{g}_{n,n+1}\\
\vdots\\
\mathbf{g}_{n,N}
\end{bmatrix}\otimes\mathbf{I}_{\frac{L}{N-T}}\Big) Z_n =\big(\mathbf{G}_{n} \otimes\mathbf{I}_{\frac{L}{N-T}}\big) Z_n, \quad \forall\,n\in[N].
\end{IEEEeqnarray}
Since $\mathbf{G}_n$ is a full-rank square matrix by \eqref{definition:Gn}, the square matrix $\mathbf{G}_{n} \otimes\mathbf{I}_{\frac{L}{N-T}}$ is also full rank, for any $n\in[N]$. Moreover, since $Z_1, Z_2, \ldots, Z_N$ are independently and uniformly generated, all the encoded keys $\{Z_{n,m}\}_{n\in[N],m\in[N]\backslash\{n\}}$ are mutually independent and uniformly distributed over $\mathbb{F}_q$.

Let $(\mathbf{G}_{n} \otimes\mathbf{I}_{\frac{L}{N-T}})^{-1}$ denote the inverse of matrix $\mathbf{G}_{n} \otimes\mathbf{I}_{\frac{L}{N-T}}$. Then, by \eqref{pairwise:1},
\begin{IEEEeqnarray}{c}\notag
Z_n=\big(\mathbf{G}_{n} \otimes\mathbf{I}_{\frac{L}{N-T}}\big)^{-1}\begin{bmatrix}
Z_{n,1}\\
\vdots\\
Z_{n,n-1}\\
Z_{n,n+1}\\
\vdots\\
Z_{n,N}
\end{bmatrix}, \quad \forall\,n\in[N].
\end{IEEEeqnarray}
Furthermore, according to \eqref{eq:maskinputlinear}, we know that the masked input $X_n$ is generated by using only the pairwise random keys $\{Z_{n,m},Z_{m,n}\}_{m\in[N]\backslash\{n\}}$ available at user $n$ to mask the private input $W_n$. Consequently, the proposed scheme can be achieved using a pairwise key-distribution structure, in which each pair of users $n$ and $m$ shares an independent common random key $Z_{n,m}$ for any $n\in[N]$ and $m\in[N]\backslash\{n\}$.
\end{IEEEproof}
This lemma completes the proof of Theorem \ref{theorem:pairwise}.

\subsection{Comparison with Pairwise-Key-Based Secure Aggregation}\label{section:Comparison}
The $T$-colluding secure aggregation problem was originally introduced in \cite{bonawitz2017practical} from the perspective of computational security, and was subsequently studied in \cite{zhao2023secure} from the information-theoretic security perspective to investigate its fundamental limits. Both works proposed $T$-colluding secure aggregation schemes under a pairwise key distribution structure. To facilitate comparison with these prior works, we briefly outline the corresponding schemes within the two-phase framework formulated in Section \ref{Problem:formulation}.

\emph{Secure Aggregation Scheme in \cite{bonawitz2017practical}:}
For any $n,m\in[N]$ with $n<m$, the pairwise key $Z_{n,m}$ shared at users $n$ and $m$ is generated by independently and uniformly drawing an $L$-dimensional vector over the finite field $\mathbb{F}_q$. Then, the pairwise keys available at user $n$ are given by $\{Z_{i,n},Z_{n,j}:1\leq i<n<j\leq N\}$ for each $n\in[N]$. Furthermore, the masked input $X_n$ sent by user $n$ to the server is designed as
\begin{IEEEeqnarray}{c}\notag
X_{n}=W_n-\sum\limits_{i=1}^{n-1}Z_{i,n}+\sum\limits_{j=n+1}^{N}Z_{n,j}.
\end{IEEEeqnarray}
After collecting the masked inputs $\{X_n\}_{n\in[N]}$ from all users, the server decodes the desired aggregation $W_{sum}$ by computing
\begin{IEEEeqnarray}{rCl}
\sum\limits_{n=1}^{N}X_n&=&\sum\limits_{n=1}^{N}W_n-\sum\limits_{n=1}^{N}\sum\limits_{i=1}^{n-1}Z_{i,n}+\sum\limits_{n=1}^{N}\sum\limits_{j=n+1}^{N}Z_{n,j} \label{comparison:1}\\
&=&\sum\limits_{n=1}^{N}W_n-\sum\limits_{1\leq i<n\leq N}Z_{i,n}+\sum\limits_{1\leq n<j\leq N}Z_{n,j} \\
&=&\sum\limits_{n=1}^{N}W_n. \label{comparison:2}
\end{IEEEeqnarray}
In general, the scheme in \cite{bonawitz2017practical} requires generating $\frac{N(N-1)}{2}$ pairwise random keys $\{Z_{n,m}:1\leq n<m\leq N\}$, each of length $L$, while each user also uploads a masked input of length $L$. According to \eqref{definition:RZ} and \eqref{definition:RKRA}, this scheme achieves a key rate of $R_Z=\frac{N(N-1)}{2}$ and a communication rate of $R_A=N$ in the aggregation phase.

\emph{Secure Aggregation Scheme in \cite{zhao2023secure}:}
Compared with the secure aggregation scheme in \cite{bonawitz2017practical}, the scheme in \cite{zhao2023secure} generates the pairwise key $Z_{n,m}$ shared between users $n$ and $m$ by independently and uniformly drawing a $\frac{2L}{N-T}$-dimensional vector from $\mathbb{F}_q$, for any $n,m\in[N]$ with $n<m$. To construct the masked input, each pairwise random key $Z_{n,m}$ is encoded using an encoding matrix $\mathbf{H}_{n,m}$ of dimensions $L\times \frac{2L}{N-T}$, whose entries are independently and uniformly chosen from a sufficiently large finite field.
Accordingly, the masked input $X_n$ sent by user $n\in[N]$ to the
server is constructed as
\begin{IEEEeqnarray}{c}\notag
X_{n}=W_n-\sum\limits_{i=1}^{n-1}\mathbf{H}_{i,n}Z_{i,n}+\sum\limits_{j=n+1}^{N}\mathbf{H}_{n,j}Z_{n,j}.
\end{IEEEeqnarray}
Similar to \eqref{comparison:1}--\eqref{comparison:2}, the server can recover the desired aggregation $W_{sum}$ by directly summing the masked inputs $\{X_n\}_{n\in[N]}$ received from all users. To ensure that such a randomized encoding construction satisfies the $T$-colluding security requirement with sufficiently high probability, the field size $q$ is required to approach infinity, while the input length $L$ exhibits super-exponential growth with the number of users $N$. Obviously, the secure aggregation scheme \cite{zhao2023secure} achieves a key rate of $R_Z=\frac{N(N-1)}{N-T}$ and a communication rate of $R_A=N$ in the aggregation phase.

Table \ref{tab:comparison} compares our proposed scheme with the secure aggregation schemes in \cite{bonawitz2017practical} and \cite{zhao2023secure}.\footnote{We omit the comparison of the communication rate in the key-distribution phase, since in practical applications pairwise keys can be established using standard cryptographic mechanisms, such as the Diffie--Hellman key-exchange protocol.} Several interesting observations and intuitive explanations are provided below.
\begin{itemize}
\item In terms of the key rate and aggregation communication rate, both our scheme and the scheme in \cite{zhao2023secure} achieve a lower key rate than the scheme in \cite{bonawitz2017practical} when $0\leq T<N-2$, while maintaining the same aggregation communication rate. When $T=N-2$, all three schemes achieve the same key rate and aggregation communication rate.
The reduction in the key rate can be attributed to two factors. First, although the scheme in \cite{bonawitz2017practical} can address the $T$-colluding secure aggregation problem, it is designed to withstand up to $N-2$ colluding users and is therefore not tailored to the specific security parameter $T$. Second, both our scheme and the scheme in \cite{zhao2023secure} are designed to match the target security level $T$ precisely, thereby reducing the amount of randomness required.
\item In terms of the required field size $q$ and input length $L$, the proposed scheme significantly improves upon the randomized coding construction in \cite{zhao2023secure}, which requires the field size to tend to infinity and the input length to grow super-exponentially with $N$. In contrast, our scheme requires only a finite field size $q\geq N$ and an input length $L=N-T$. This improvement stems from our explicit deterministic construction, which is obtained by first establishing a general linear coding framework and then explicitly realizing the framework through careful exploitation of the algebraic structure of $N\times(N-T)$ Vandermonde matrices. The structural properties of these Vandermonde matrices enable us to reduce the required field size to $N$ and the required input length to $N-T$. Moreover, compared with the scheme in \cite{bonawitz2017practical}, the proposed scheme requires a larger field size and a longer input length due to the structural constraints imposed by the Vandermonde-matrix-based construction.
\item It is also worth noting that although reference \cite{zhao2023secure} characterizes the optimal key rate and aggregation communication rate for the $T$-colluding secure aggregation problem, the characterization is restricted to the symmetric groupwise key-distribution structure, in which every user subset of size $2\leq G\leq N-T$ shares an independent common random key. For more general key-distribution models, the fundamental limits of the secure aggregation problem remain unknown. 
In contrast to \cite{zhao2023secure}, we investigate the secure aggregation problem under a general two-phase framework that jointly incorporates the key-distribution phase and the aggregation phase, where correlated keys are established through user-to-user communication. Under this general framework, we completely characterize the optimal key rate and the optimal aggregation communication rate. Moreover, we show that a pairwise key-distribution structure is sufficient for achieving these fundamental limits, thereby establishing the optimality of a remarkably simple key-distribution architecture.
\end{itemize}


\begin{table}[htbp]
\centering
\caption{Performance comparison between our proposed scheme and the secure aggregation schemes in \cite{bonawitz2017practical} and \cite{zhao2023secure} under a pairwise key distribution structure.}
\label{tab:comparison}
\begin{NiceTabular}{lccc}[hvlines]
\rule[-3.0mm]{0pt}{8mm}\diagbox{
  \raisebox{1mm}{\hspace{1.5mm}Metrics}
}{
  \raisebox{0mm}{Secure Aggregation Schemes\hspace{0mm}}
}
&  Google's Scheme \cite{bonawitz2017practical} & Zhao-Sun Scheme \cite{zhao2023secure} & Our Scheme \\
Key Rate $R_Z$ & $\frac{N(N-1)}{2}$ & $\frac{N(N-1)}{N-T}$ & $\frac{N(N-1)}{N-T}$ \\ 
Aggregation-Communication Rate $R_A$ & $N$ & $N$ & $N$ \\ 
Required Finite Field $\mathbb{F}_q$ & $q\geq 2$ & $q=p^M$, $p$ is any prime number, $M\rightarrow\infty$ & $q\geq N$ \\
Minimum Required Input Length $L$ & $1$ & $N!\binom{N-T}{2}$ & $N-T$ \\
Construction Type & Explicit & Probabilistic & Explicit \end{NiceTabular}
\end{table}


\section{Proof of Theorem \ref{theo:capacity}: Converse}\label{Converse Proof}
In this section, we prove the converse bounds of Theorem \ref{theo:capacity}. 
We start with two useful lemmas that will be instrumental in establishing the converse bounds. To facilitate a better understanding of these lemmas, we first explain their implications and provide the intuition behind their proofs before formally presenting them.

The following lemma shows that for any user $n\in[N]$, the local input $W_n$ can be fully recovered from the masked input $X_n$, the encoded keys $\{Z_{n,m}\}_{m\in[N]\backslash\{n\}}$ sent to the other users, and the encoded keys $\{Z_{m,n}\}_{m\in[N]\backslash\{n\}}$ received from the other users. According to \eqref{encodedkey} and \eqref{model:maskedinput}, we know that user $n$ locally holds the random variables $Z_n,\{Z_{n,m}\}_{m\in[N]\backslash\{n\}}$, and $\{Z_{m,n}\}_{m\in[N]\backslash\{n\}}$, which are intended to protect the privacy of $W_n$. 
In fact, this lemma implies that the local input $W_n$ is effectively masked only by the randomness contained in $\{Z_{n,m}\}_{m\in[N]\backslash\{n\}}$ and $\{Z_{m,n}\}_{m\in[N]\backslash\{n\}}$, i.e., the remaining randomness in $Z_n$ beyond $\{Z_{n,m}\}_{m\in[N]\backslash\{n\}}$ does not contribute to masking $W_n$.
This is because such remaining randomness is never transmitted to the other users; if it were used to mask $W_n$, the resulting interference could not be canceled, thereby violating the correctness constraint in \eqref{model:correctness}.
\begin{Lemma}\label{lemma:encodedkey-determined-masking}
For any given $n\in[N]$,
    \begin{IEEEeqnarray}{c}
        H\big(W_n| X_n, \{Z_{n,m}, Z_{m,n}\}_{m\in[N]\backslash\{n\}}\big) = 0. \notag
    \end{IEEEeqnarray}
\end{Lemma}
\begin{IEEEproof}
For any $n\in[N]$, we have
\begin{IEEEeqnarray}{rCl}
0&\leq&I\big(W_n;\{Z_m\}_{m\in[N]\backslash\{n\}} | X_n,\{Z_{n,m},Z_{m,n}\}_{m\in[N]\backslash\{n\}} \big) \notag\\
&\leq& I\big(W_n,X_n,Z_n;\{Z_m\}_{m\in[N]\backslash\{n\}} |\{Z_{n,m},Z_{m,n}\}_{m\in[N]\backslash\{n\}} \big) \notag\\
&\overset{(a)}{=}& I\big(W_n,Z_n;\{Z_m\}_{m\in[N]\backslash\{n\}} |\{Z_{n,m},Z_{m,n}\}_{m\in[N]\backslash\{n\}} \big) \notag\\
 &\leq&I\big(W_n,Z_n,\{Z_{n,m}\}_{m\in[N]\backslash\{n\}};\{Z_m,Z_{m,n}\}_{m\in[N]\backslash\{n\}} \big) \notag\\
&\overset{(b)}{=}&I\big(W_n,Z_n;\{Z_m\}_{m\in[N]\backslash\{n\}} \big) \notag\\
&\overset{(c)}{=}&0, \notag
\end{IEEEeqnarray}
where $(a)$ follows by \eqref{model:maskedinput}; $(b)$ follows from the fact that $Z_{n,m}$ is determined by $Z_{n}$ for any $n\in[N]$ and $m\in[N]\backslash\{n\}$ by \eqref{encodedkey}; $(c)$ is due to the fact that $\{Z_m\}_{m\in[N]\backslash\{n\}}$ is generated independently of $W_n$ and $Z_n$ by \eqref{model:key:inden}. Thus, we obtain
\begin{IEEEeqnarray}{c}
I\big(W_n;\{Z_m\}_{m\in[N]\backslash\{n\}} | X_n,\{Z_{n,m},Z_{m,n}\}_{m\in[N]\backslash\{n\}} \big)=0, \notag
\end{IEEEeqnarray}
which is employed to complete the proof of the lemma as follows:
\begin{IEEEeqnarray}{rCl}
0&\leq&H\big(W_n | X_n,\{Z_{n,m},Z_{m,n}\}_{m\in[N]\backslash\{n\}} \big)\nonumber\\
&=&H\big(W_n | X_n,\{Z_{n,m},Z_{m,n}\}_{m\in[N]\backslash\{n\}} \big)-I\big(W_n;\{Z_m\}_{m\in[N]\backslash\{n\}} | X_n,\{Z_{n,m},Z_{m,n}\}_{m\in[N]\backslash\{n\}} \big)\nonumber\\
&=&H\big(W_n | X_n,\{Z_{n,m},Z_{m,n},Z_m\}_{m\in[N]\backslash\{n\}} \big)\nonumber\\ 
&\overset{(a)}{=}&H\big(W_n | X_n,\{Z_{n,m},Z_m\}_{m\in[N]\backslash\{n\}} \big)\nonumber\\ 
&\overset{(b)}{=}&H\big(W_n | X_n,\{W_m,Z_{n,m},Z_m\}_{m\in[N]\backslash\{n\}} \big)\nonumber\\ 
&=&H\big(W_{sum} | X_n,\{W_m,Z_m,Z_{n,m}\}_{m\in[N]\backslash\{n\}} \big)\nonumber\\
&\stackrel{(c)}{=}&H\big(W_{sum} | \{X_m\}_{m\in[N]},\{W_m,Z_m,Z_{n,m}\}_{m\in[N]\backslash\{n\}} \big)\nonumber\\
&\leq&H\big(W_{sum} | \{X_m\}_{m\in[N]} \big)\nonumber\\
&\stackrel{(d)}{=}&0,\notag
\end{IEEEeqnarray}
where $(a)$ is due to \eqref{encodedkey}; $(b)$ follows from the fact that $X_n$ is a deterministic function of $\{W_{n},Z_n,\{Z_{m,n}\}_{m\in[N]\backslash\{n\}}\}$ by \eqref{model:maskedinput}, and that $\{W_m\}_{m\in[N]\backslash\{n\}}$ is independent of 
$\{W_{n},Z_n,\{Z_{m,n},Z_{n,m},Z_m\}_{m\in[N]\backslash\{n\}}\}$ by  \eqref{encodedkey} and \eqref{model:key:inden}; 
$(c)$ follows by the fact that $X_m$ can be determined from $W_{m},Z_m,\{Z_{\ell}\}_{\ell\in[N]\backslash\{m,n\}}$, and $Z_{n,m}$ for any $m\in[N]\backslash\{n\}$ by \eqref{encodedkey} and \eqref{model:maskedinput}; $(d)$ follows from the correctness constraint in \eqref{model:correctness}.
\end{IEEEproof}

The following lemma establishes a lower bound on the size of the encoding keys exchanged among the non-colluding users, conditioned on the encoding keys sent from the non-colluding users to the colluding users.
For any given subset of colluding users $\mathcal{T}\subseteq[N]$ of size $T$, the proof of the lemma 
is completed by jointly exploiting the following two key observations: 
1) The security constraint in \eqref{model:security} guarantees that an adversary who eavesdrops on the server and the colluding users in $\mathcal{T}$ obtains no information about the inputs $\{W_n\}_{n\in[N]\backslash\mathcal{T}}$ of the remaining $N-T$ non-colluding users in $[N]\backslash\mathcal{T}$, except for their aggregated result $\sum_{n\in[N]\backslash\mathcal{T}}W_n$. From an information-theoretic perspective, concealing all information about the $N-T$ individual inputs beyond their sum requires randomness of size at least $(N-T-1)L$. Hence, this requires that, conditioned on all the randomness held by the colluding users in $\mathcal{T}$, the total amount of randomness among the non-colluding users in $[N]\backslash\mathcal{T}$ must be no less than $(N-T-1)L$; otherwise, the security constraint in \eqref{model:security} would be violated.
2) The input $W_n$ of each user $n\in[N]$ is masked exclusively by the encoded keys $\{Z_{n,m}\}_{m\in[N]\backslash\{n\}}$ sent to the other users and the encoded keys $\{Z_{m,n}\}_{m\in[N]\backslash\{n\}}$ received from the other users, as demonstrated in Lemma \ref{lemma:encodedkey-determined-masking}.
\begin{Lemma}\label{lemma:non-colludingkey-entropy}
    For any subset of colluding users $\mathcal{T}\subseteq[N]$ with size $|\mathcal{T}|=T$ and its complement $\mathcal{T}^c=[N]\backslash{\mathcal{T}}$, 
    \begin{IEEEeqnarray}{rCl}
        H\big(\{Z_{n,m}\}_{n \in  \mathcal{T}^c,m \in \mathcal{T}^c\backslash\{n\}}|\{Z_{n,m}\}_{n\in\mathcal{T}^c,m \in \mathcal{T}}\big) &\ge& (N-T-1)L.\label{eq:lemma:non-colludingencodedkey-entropy}
    \end{IEEEeqnarray}
\end{Lemma}
\begin{IEEEproof}
For any given subset $\mathcal{T}\subseteq[N]$ with size $|\mathcal{T}|=T$ and its complement $\mathcal{T}^c=[N]\backslash\mathcal{T}$, we have
\begin{IEEEeqnarray*}{rCl}
&&H\big(\{W_n\}_{n\in\mathcal{T}^c} | \{Z_{n,m}\}_{n \in  \mathcal{T}^c,m \in \mathcal{T}^c\backslash\{n\}},
\{Z_{n,m}\}_{n\in\mathcal{T}^c,m\in\mathcal{T}}, \{X_{n}\}_{n\in\mathcal{T}^c},\{Z_{n}\}_{n\in\mathcal{T}}\big)\\
&\overset{(a)}{=}&H\big(\{W_n\}_{n\in\mathcal{T}^c} | \{Z_{n,m}\}_{n \in  \mathcal{T}^c,m \in \mathcal{T}^c\backslash\{n\}}, 
\{Z_{n,m}\}_{n\in\mathcal{T}^c,m\in\mathcal{T}},\{X_{n}\}_{n\in\mathcal{T}^c},\{Z_{m}\}_{m\in\mathcal{T}},\{Z_{m,n}\}_{m\in\mathcal{T},n\in\mathcal{T}^c}\big)\\
&\overset{(b)}{=}&H\big(\{W_n\}_{n\in\mathcal{T}^c} | \{Z_{n,m}\}_{n \in  \mathcal{T}^c,m \in [N]\backslash\{n\}},
\{X_{n}\}_{n\in\mathcal{T}^c},\{Z_{m}\}_{m\in\mathcal{T}},\{Z_{m,n}\}_{m\in[N]\backslash\{n\},n\in\mathcal{T}^c}\big)\\
&\leq&H\big(\{W_n\}_{n\in\mathcal{T}^c} | \{Z_{n,m}\}_{n \in  \mathcal{T}^c,m \in [N]\backslash\{n\}},
\{X_{n}\}_{n\in\mathcal{T}^c},\{Z_{m,n}\}_{m\in[N]\backslash\{n\},n\in\mathcal{T}^c}\big)\\
&\overset{(c)}{=}&0,
\end{IEEEeqnarray*}
where $(a)$ follows from the fact that $Z_{m,n}$ is a deterministic function of $Z_{m}$ for any $m\in\mathcal{T}$ and $n\in\mathcal{T}^c$ by \eqref{encodedkey}; $(b)$ follows by the fact that 
$\{Z_{m,n}\}_{m\in[N]\backslash\{n\},n\in\mathcal{T}^c}=\{Z_{n,m}\}_{n \in  \mathcal{T}^c,m \in \mathcal{T}^c\backslash\{n\}}\bigcup\{Z_{m,n}\}_{m\in\mathcal{T},n\in\mathcal{T}^c}$; $(c)$ is due to Lemma \ref{lemma:encodedkey-determined-masking}. 
Therefore, we obtain
\begin{IEEEeqnarray}{c}\label{lemma:auxiliary:result:1}
H\big(\{W_n\}_{n\in\mathcal{T}^c} | \{Z_{n,m}\}_{n \in  \mathcal{T}^c,m \in \mathcal{T}^c\backslash\{n\}},
\{Z_{n,m}\}_{n\in\mathcal{T}^c,m\in\mathcal{T}}, \{X_{n}\}_{n\in\mathcal{T}^c},\{Z_{n}\}_{n\in\mathcal{T}}\big)=0. 
\end{IEEEeqnarray}
Following \eqref{lemma:auxiliary:result:1}, we continue
the proof of the lemma.
\begin{IEEEeqnarray*}{rCl}
&&H\big(\{Z_{n,m}\}_{n \in \mathcal{T}^c,m \in \mathcal{T}^c\backslash\{n\}} | \{Z_{n,m}\}_{n\in\mathcal{T}^c,m \in \mathcal{T}}\big)\\
&\geq&H\big(\{Z_{n,m}\}_{n \in  \mathcal{T}^c,m \in \mathcal{T}^c\backslash\{n\}} | 
\{Z_{n,m}\}_{n\in\mathcal{T}^c,m\in\mathcal{T}}, \{X_{n}\}_{n\in\mathcal{T}^c},\{Z_{n}\}_{n\in\mathcal{T}}\big)\\
&\overset{(a)}{=}&H\big(\{Z_{n,m}\}_{n \in  \mathcal{T}^c,m \in \mathcal{T}^c\backslash\{n\}} | 
\{Z_{n,m}\}_{n\in\mathcal{T}^c,m\in\mathcal{T}}, \{X_{n}\}_{n\in\mathcal{T}^c},\{Z_{n}\}_{n\in\mathcal{T}}\big)\\
&&+H\big(\{W_n\}_{n\in\mathcal{T}^c} | \{Z_{n,m}\}_{n \in  \mathcal{T}^c,m \in \mathcal{T}^c\backslash\{n\}},
\{Z_{n,m}\}_{n\in\mathcal{T}^c,m\in\mathcal{T}}, \{X_{n}\}_{n\in\mathcal{T}^c},\{Z_{n}\}_{n\in\mathcal{T}}\big)\\
&=&H\big(\{W_n\}_{n\in\mathcal{T}^c},\{Z_{n,m}\}_{n \in  \mathcal{T}^c,m \in \mathcal{T}^c\backslash\{n\}}|
\{Z_{n,m}\}_{n\in\mathcal{T}^c,m\in\mathcal{T}}, \{X_{n}\}_{n\in\mathcal{T}^c},\{Z_{n}\}_{n\in\mathcal{T}}\big)\\
&\geq&H\big(\{W_n\}_{n\in\mathcal{T}^c}|
\{Z_{n,m}\}_{n\in\mathcal{T}^c,m\in\mathcal{T}}, \{X_{n}\}_{n\in\mathcal{T}^c},\{Z_{n}\}_{n\in\mathcal{T}}\big)\\
&\geq&H\big(\{W_n\}_{n\in\mathcal{T}^c}| \sum\limits_{n\in \mathcal{T}^c} W_n,
\{Z_{n,m}\}_{n\in\mathcal{T}^c,m\in\mathcal{T}}, \{X_{n}\}_{n\in\mathcal{T}^c},\{Z_{n}\}_{n\in\mathcal{T}}\big)\\
&=&H\big(\{W_n\}_{n\in\mathcal{T}^c}| \sum\limits_{n\in \mathcal{T}^c} W_n\big)-I\big(\{W_n\}_{n\in\mathcal{T}^c};\{Z_{n,m}\}_{n\in\mathcal{T}^c,m\in\mathcal{T}}, \{X_{n}\}_{n\in\mathcal{T}^c},\{Z_{n}\}_{n\in\mathcal{T}}|\sum\limits_{n\in \mathcal{T}^c} W_n\big)\\
&\overset{(b)}{=}&H\big(\{W_n\}_{n\in\mathcal{T}^c}| \sum\limits_{n\in \mathcal{T}^c} W_n\big)-I\big(\{W_n\}_{n\in\mathcal{T}^c};\{Z_{n,m}\}_{n\in\mathcal{T}^c,m\in\mathcal{T}}, \{X_{n}\}_{n\in\mathcal{T}^c},\{Z_{n}\}_{n\in\mathcal{T}}|\sum\limits_{n\in \mathcal{T}^c} W_n,\{W_n\}_{n\in\mathcal{T}}\big)\\
&\overset{(c)}{=}&H\big(\{W_n\}_{n\in\mathcal{T}^c}| \sum\limits_{n\in \mathcal{T}^c} W_n\big) \notag\\
&\overset{(d)}{=}&(N-T-1)L,
\end{IEEEeqnarray*} 
where $(a)$ is due to \eqref{lemma:auxiliary:result:1};
$(b)$ follows from the fact that $X_n$ can be determined by $\{W_{n},\{Z_n\}_{n\in[N]}\}$ for any $n\in\mathcal{T}^c$ by \eqref{encodedkey}-\eqref{model:maskedinput}, and $\{W_n\}_{n\in\mathcal{T}}$ is independent of $\{\{W_n\}_{n\in\mathcal{T}^c},\{Z_{n,m}\}_{n\in\mathcal{T}^c,m\in\mathcal{T}},\{Z_{n}\}_{n\in[N]},\sum_{n\in \mathcal{T}^c} W_n\}$ by \eqref{encodedkey} and \eqref{model:key:inden};
$(c)$ follows by the security constraint in \eqref{model:security};
$(d)$ holds because $W_1,\ldots,W_N$ are independent and uniformly distributed length-$L$ random vectors over the finite field $\mathbb{F}_q$ by
\eqref{model:file:inden} and \eqref{infor:indenpe}.
 \end{IEEEproof}

Next, we present a formal proof of the converse bounds of Theorem \ref{theo:capacity}.

\subsection{Proof of $R_A \geq N$}
In this subsection, we complete the converse proof for the aggregation communication rate $R_A$. As a first step, we focus on the size of the masked input of an individual user and derive a lower bound on the masked input $X_n$ for any given $n\in[N]$ as follows:
\begin{IEEEeqnarray*}{rCl}
H(X_n)&\geq& H(X_n|\{X_m\}_{m\in[N]\backslash\{n\}}) \nonumber \\
&\ge& H\big(X_n|\{X_m\}_{m\in[N]\backslash\{n\}}\big)-H\big(X_n|W_{sum}, \{X_m\}_{m\in[N]\backslash\{n\}}\big)\\
&=& I\big(X_n;W_{sum}|\{X_m\}_{m\in[N]\backslash\{n\}}\big)\\
& = & H\big(W_{sum}|\{X_m\}_{m\in[N]\backslash\{n\}}\big)-H\big(W_{sum}|\{X_m\}_{m\in[N]}\big) \\
&\stackrel{(a)}{=} & H\big(W_{sum}|\{X_m\}_{m\in[N]\backslash\{n\}}\big)\\
&\geq&H\big(W_{sum}|\{X_m\}_{m\in[N]\backslash\{n\}},\{W_m\}_{m\in[N]\backslash\{n\}},\{Z_m\}_{m\in[N]\backslash\{n\}},\{Z_{\ell,m}\}_{\ell\in[N]\backslash\{m\},m\in[N]\backslash\{n\}}\big)\\
&\stackrel{(b)}{=}&H\big(W_{sum}|\{W_m\}_{m\in[N]\backslash\{n\}},\{Z_m\}_{m\in[N]\backslash\{n\}},\{Z_{\ell,m}\}_{\ell\in[N]\backslash\{m\},m\in[N]\backslash\{n\}}\big)\\
&=&H\big(W_{n}|\{W_m\}_{m\in[N]\backslash\{n\}},\{Z_m\}_{m\in[N]\backslash\{n\}},\{Z_{\ell,m}\}_{\ell\in[N]\backslash\{m\},m\in[N]\backslash\{n\}}\big)\\
&\stackrel{(c)}{=} &H\big(W_{n}\big) \\
&\overset{(d)}{=}&L,
\end{IEEEeqnarray*}
where 
$(a)$ follows by the correctness constraint in \eqref{model:correctness}; 
$(b)$ follows from the fact that $X_m$ is a deterministic function of $W_m,Z_m$, and $\{Z_{\ell,m}\}_{\ell\in[N]\backslash\{m\}}$ for any $m\in[N]\backslash\{n\}$ by \eqref{model:maskedinput}; 
$(c)$ is due to the fact that $\{Z_{\ell,m}\}_{\ell\in[N]\backslash\{m\},m\in[N]\backslash\{n\}}$ can be generated from $\{Z_{m}\}_{m\in[N]}$ by  \eqref{encodedkey}, and that $W_{n}$ is independent of $\{W_m\}_{m\in[N]\backslash\{n\}}$ and $\{Z_m\}_{m\in[N]}$ by \eqref{model:key:inden};
$(d)$ follows by \eqref{infor:indenpe}. 
This result is intuitive, since each user $n\in[N]$ must transmit a message of length at least equal to that of its local input $W_n$; otherwise, the desired aggregation $W_{sum}$ cannot be recovered from $\{X_n\}_{n\in[N]}$, and the correctness constraint in \eqref{model:correctness} would be violated.

Then, according to the definition of $R_A$ in \eqref{definition:RKRA}, we establish the converse bound for $R_A$ as follows:
    \begin{IEEEeqnarray}{C}\label{converse:RA}
R_A=\frac{\sum_{n\in[N]}H(X_n)}{H(W_{sum})}\geq\frac{NL}{L}=N.
    \end{IEEEeqnarray}

\subsection{Proof of $R_K \geq \frac{N (N-1)}{N - T}$ and $R_Z \geq \frac{N (N-1)}{N - T}$}
In this subsection, we prove the converse bounds for the key-distribution communication rate $R_K$ and the key rate $R_Z$. This is accomplished by applying Lemma \ref{lemma:non-colludingkey-entropy} to all possible colluding subsets $\mathcal{T}\subseteq[N]$ of size $|\mathcal{T}|=T$ and summing the resulting inequalities.

For any subset of colluding users $\mathcal{T}\subseteq[N]$ with size $|\mathcal{T}|=T$, let $\mathcal{T}^c$ denote its complement, i.e., $\mathcal{T}^c=[N]\backslash\mathcal{T}$. 
Since the inequality in \eqref{eq:lemma:non-colludingencodedkey-entropy} holds for all such subsets $\mathcal{T}$, we can sum it over all $\mathcal{T}\subseteq[N]$ with $|\mathcal{T}|=T$ to obtain 
\begin{IEEEeqnarray}{rCl}
\Big(\substack{ N\\[1ex] T } \Big)(N-T-1)L&\le&\sum_{\substack{\mathcal{T}\subseteq [N]\\|\mathcal{T}| = T}} H\big(\{Z_{n,m}\}_{n\in  \mathcal{T}^c,m\in \mathcal{T}^c\backslash\{n\}} | \{Z_{n,m}\}_{n\in\mathcal{T}^c,m \in \mathcal{T}}\big) \nonumber\\
&=&\sum_{\substack{\mathcal{T}\subseteq [N]\\|\mathcal{T}| = T}} H\big(\{Z_{n,m}\}_{n\in  \mathcal{T}^c,m \in [N]\backslash\{n\}} \big)-\sum_{\substack{\mathcal{T}\subseteq [N]\\|\mathcal{T}| = T}} H\big(\{Z_{n,m}\}_{n\in\mathcal{T}^c,m \in \mathcal{T}} \big) \nonumber\\
&\stackrel{(a)}{=}&\sum_{\substack{\mathcal{T}\subseteq [N]\\|\mathcal{T}| = T}} \sum_{n \in  \mathcal{T}^c}H\big(\{Z_{n,m}\}_{m \in [N]\backslash\{n\}} \big)-\sum_{\substack{\mathcal{T}\subseteq [N]\\|\mathcal{T}| = T}} \sum_{n\in\mathcal{T}^c} H\big(\{Z_{n,m}\}_{m\in \mathcal{T}} \big) \nonumber\\
&=&\Big( \substack{N-1 \\[1ex] T} \Big)\sum_{n\in[N]}H\big(\{Z_{n,m}\}_{m\in [N]\backslash\{n\}} \big)-
\sum_{\substack{\mathcal{T}\subseteq [N]\\|\mathcal{T}| = T}} \sum_{n\in\mathcal{T}^c} H\big(\{Z_{n,m}\}_{m\in \mathcal{T}} \big)\nonumber\\
&\stackrel{(b)}{=}&\Big( \substack{N-1 \\[1ex] T} \Big)\sum_{n\in[N]}H\big(\{Z_{n,m}\}_{m\in [N]\backslash\{n\}} \big)
-\sum_{\substack{\mathcal{T}\subseteq [N]\\|\mathcal{T}| = T}} 
\sum_{n \in \mathcal{T}^c}
\sum_{m \in \mathcal{T}} H\big(Z_{n,m}|\{Z_{n,\ell}\}_{\ell \in \mathcal{T},\ell <m}\big)\nonumber\\
&\stackrel{(c)}{\le}&\Big( \substack{N-1 \\[1ex] T} \Big)\sum_{n\in[N]}H\big(\{Z_{n,m}\}_{m\in [N]\backslash\{n\}} \big)
-
\sum_{\substack{\mathcal{T}\subseteq [N]\\|\mathcal{T}| = T}} 
\sum_{n \in \mathcal{T}^c}
\sum_{m \in \mathcal{T}} H\big(Z_{n,m}| \{Z_{n,\ell}\}_{\ell \in [N]\backslash\{n\},\ell < m}\big)\nonumber\\
&=&\Big( \substack{N-1 \\[1ex] T} \Big)\sum_{n\in[N]}H\big(\{Z_{n,m}\}_{m\in [N]\backslash\{n\}} \big)
- 
\Big( \substack{N-2 \\[1ex] T-1} \Big)\sum_{n \in [N]}
\sum_{m \in [N]\backslash\{n\}} 
H\big(Z_{n,m}| \{Z_{n,\ell}\}_{\ell \in [N]\backslash\{n\},\ell < m}\big)\nonumber\\
&\stackrel{(d)}{=}&\Big( \substack{N-1 \\[1ex] T} \Big)\sum_{n\in[N]}H\big(\{Z_{n,m}\}_{m\in [N]\backslash\{n\}} \big)
- 
\Big( \substack{N-2 \\[1ex] T-1} \Big) \sum_{n\in [N]}H\big(\{Z_{n,m}\}_{m \in [N]\backslash\{n\}} \big)\nonumber\\
&=&\Big( \substack{N-2 \\[1ex] T} \Big)\sum_{n\in[N]}H\big(\{Z_{n,m}\}_{m\in [N]\backslash\{n\}} \big), \notag
\end{IEEEeqnarray}
where $(a)$ follows from the fact that $Z_{n,m}$ is a deterministic function of $Z_{n}$ for any $n\in\mathcal{T}^c$ and $m\in[N]\backslash\{n\}$ by \eqref{encodedkey}, and that $\{Z_n\}_{n\in\mathcal{T}^c}$ are mutually independent by \eqref{model:key:inden};
$(b)$ and $(d)$ follow by applying the chain rule of joint entropy;
$(c)$ holds because the conditioning does not increase entropy. 
Therefore, we obtain
\begin{IEEEeqnarray*}{c}
\sum_{n \in  [N]}H\big(\{Z_{n,m}\}_{m \in [N]\backslash\{n\}} \big) \geq \frac{ \binom{N}{T}(N-T-1)}{\binom{N-2}{T}}L= \frac{N (N-1)}{N - T}L.
\end{IEEEeqnarray*}

Then, according to the definition of $R_K$ in \eqref{definition:RKRA}, we establish the converse bound for $R_K$ as follows:
\begin{IEEEeqnarray}{rCl}
R_K&=&\frac{\sum_{n\in[N]}\sum_{m\in[N]\backslash\{n\}}H(Z_{n,m})}{H(W_{sum})}\nonumber\\
&\ge&\frac{\sum_{n \in  [N]}H\big(\{Z_{n,m}\}_{m \in [N]\backslash\{n\}} \big)}{L}\nonumber\\
&\geq&\frac{N (N-1)}{N - T}. 
\label{converse:RK}
\end{IEEEeqnarray}
Similarly, from the definition of $R_Z$ in \eqref{definition:RZ}, the converse bound for $R_Z$ can be derived as follows:
\begin{IEEEeqnarray}{rCl}
R_Z &=&\frac{\sum_{n\in[N]}H(Z_n)}{H(W_{sum})} \notag \\
&\stackrel{(a)}{\ge}&\frac{\sum_{n\in[N]}H(\{Z_{n,m}\}_{m\in[N]\backslash\{n\}})}{L} \notag \\
&\geq&\frac{N (N-1)}{N - T},\label{converse:RZ}
\end{IEEEeqnarray}
where $(a)$ holds since the collection of encoded keys $\{Z_{n,m}\}_{m\in[N]\backslash\{n\}}$ is a deterministic function of $Z_n$ for any $n\in[N]$ by \eqref{encodedkey}. 

From \eqref{converse:RA}, \eqref{converse:RK}, and \eqref{converse:RZ}, we have shown that any achievable rate tuple $(R_Z, R_K, R_A)$ for the secure aggregation problem is lower bounded by $\big(\frac{N (N-1)}{N - T},\frac{N (N-1)}{N - T},N\big)$, thereby completing the converse proof of Theorem \ref{theo:capacity}.

\section{Conclusion}\label{conclusion}
In this paper, we considered the $T$-colluding information-theoretic secure aggregation problem under a unified two-phase formulation consisting of a key distribution phase and an update aggregation phase. 
Unlike most existing works, which either rely on a trusted third party or impose prescribed symmetric groupwise key-distribution structures, we considered a general  key-distribution framework in which correlated random keys are established through  user-to-user communication. 
Under this general formulation, we completely characterized the capacity region in terms of three resources: the amount of random keys required for security, the communication required for key distribution, and the communication required for update aggregation. 
On the achievability side, we developed an explicit capacity-achieving secure aggregation scheme based on a novel linear coding design. 
Furthermore, we showed that the optimal performance can already be achieved using only a pairwise key-distribution structure, in which every pair of users shares a mutually independent common random key. 
This structural result bridges the gap between information-theoretic secure aggregation and practical cryptographic implementations based on pairwise key-establishment protocols such as Diffie--Hellman key exchange.
The converse result applied to general user-to-user key-distribution mechanisms and therefore established that pairwise key distribution is an information-theoretically optimal structure for designing capacity-achieving secure aggregation schemes. 
In contrast to most existing secure aggregation schemes that either rely on trusted third parties or employ randomized or existential constructions over sufficiently large finite fields, the proposed scheme provides an explicit deterministic construction over any finite field whose size grows linearly with the number of users.
One important on-going work is to extend the proposed framework and converse techniques to more sophisticated secure aggregation settings, such as secure aggregation with user dropouts, hierarchical aggregation architectures, and decentralized distributed learning systems.

\begin{appendix}\label{public elements}

In this appendix, we provide the formal proof of Proposition \ref{pro:condition}.
Since the performance of the proposed coding framework matches the optimal rate tuple $\big(\frac{N (N-1)}{N - T},\frac{N (N-1)}{N - T},N\big)$, 
the proposition can be established by proving the following two aspects:
1) if the encoding matrices 
$\{\mathbf{E}_{m,n}\}_{m\in[N],n\in[N]}$ and $\{\mathbf{P}_{m,n},\mathbf{G}_{m,n}\}_{m\in[N]\backslash\{n\},n\in[N]}$ satisfy the conditions in \eqref{eq:zerosum-constraint} and \eqref{eq:rankcondition}, then the constructed coding framework simultaneously satisfies the correctness constraint in \eqref{model:correctness} and the security constraint in \eqref{model:security}; and
2) when the inputs $W_1,\ldots,W_N$ are independent and uniformly distributed over the finite field $\mathbb{F}_q$, the coding framework that satisfies the correctness constraint in \eqref{model:correctness} and the security constraint in \eqref{model:security} must necessarily satisfy the conditions in \eqref{eq:zerosum-constraint} and \eqref{eq:rankcondition}.

We first show that the condition in \eqref{eq:zerosum-constraint} is sufficient and necessary for the coding framework to satisfy the correctness constraint in \eqref{model:correctness}.
Let $\mathbf{E}_{n}$ denote the $n$-th row block of the matrix $\mathbf{E}$, i.e.,
\begin{IEEEeqnarray*}{c}
\mathbf{E}_{n}=
\begin{bmatrix}
\mathbf{E}_{1,n}&\mathbf{E}_{2,n}&\cdots&\mathbf{E}_{N,n}
\end{bmatrix}, \quad\forall\, n\in[N].
\end{IEEEeqnarray*}
Then, for the masked inputs $\{X_n\}_{n\in[N]}$ in the coding framework, we have  
\begin{IEEEeqnarray}{rCl}
H\big(W_{sum} |\{X_n\}_{n\in[N]} \big)&\stackrel{(a)}{=}&H\Big(\sum_{n=1}^{N}X_n-\sum_{n=1}^N \sum_{m=1}^N\mathbf{E}_{m,n}Z_m |\{X_n\}_{n\in[N]}\Big)\nonumber\\
&=&H\Big(\sum_{n=1}^N \sum_{m=1}^N\mathbf{E}_{m,n}Z_m|\{X_n\}_{n\in[N]}\Big)
\nonumber\\
&=&H\Big(\{X_n\}_{n\in[N]},\sum_{n=1}^N \sum_{m=1}^N\mathbf{E}_{m,n}Z_m\Big)-H\Big(\{X_n\}_{n\in[N]}\Big)\nonumber\\
&\stackrel{(b)}{=}&
H\Big(
\begin{bmatrix}
\mathbf{I}_{NL}&\mathbf{E}\\
\mathbf{0}_{L\times NL}&\sum_{n=1}^N \mathbf{E}_{n}
\end{bmatrix}
\begin{bmatrix}
W_{1}\\
\vdots\\
W_{N}\\
Z_{1}\\
\vdots\\
Z_N\\
\end{bmatrix}
\Big)-H\Big(
\begin{bmatrix}
\mathbf{I}_{NL}&\mathbf{E}\\
\end{bmatrix}
\begin{bmatrix}
W_{1}\\
\vdots\\
W_{N}\\
Z_{1}\\
\vdots\\
Z_N\\
\end{bmatrix} \Big),
\label{eq:correctnesseq:1}
\end{IEEEeqnarray}
where $(a)$ follows from \eqref{eq:encodingscheme} such that $W_{sum}=\sum_{n=1}^{N}W_n=\sum_{n=1}^{N}X_n-\sum_{n=1}^{N}\sum_{m=1}^{N}\mathbf{E}_{m,n}Z_m$, and $(b)$ follows by \eqref{eq:encodingscheme} again.
Obviously, if the condition in \eqref{eq:zerosum-constraint} holds, then the coding framework satisfies $H(W_{sum} |\{X_n\}_{n\in[N]})=0$, which matches the correctness constraint in \eqref{model:correctness}. Therefore, the condition in \eqref{eq:zerosum-constraint} is sufficient for the coding framework to satisfy the correctness constraint in \eqref{model:correctness}. 

Furthermore, when the inputs $W_1,\ldots,W_N$ are independently and uniformly distributed, the equation in \eqref{eq:correctnesseq:1} can be derived as follows: 
\begin{IEEEeqnarray}{rCl}
H\big(W_{sum} |\{X_n\}_{n\in[N]} \big)
&=&
H\Big(
\begin{bmatrix}
\mathbf{I}_{NL}&\mathbf{E}\\
\mathbf{0}_{L\times NL}&\sum_{n=1}^N \mathbf{E}_{n}
\end{bmatrix}
\begin{bmatrix}
W_{1}\\
\vdots\\
W_{N}\\
Z_{1}\\
\vdots\\
Z_N\\
\end{bmatrix}
\Big)-H\Big(
\begin{bmatrix}
\mathbf{I}_{NL}&\mathbf{E}\\
\end{bmatrix}
\begin{bmatrix}
W_{1}\\
\vdots\\
W_{N}\\
Z_{1}\\
\vdots\\
Z_N\\
\end{bmatrix} \Big) \notag\\
&\stackrel{(a)}{=}&\operatorname{rank}\Big(
\begin{bmatrix}
\mathbf{I}_{NL}&\mathbf{E}\\
\mathbf{0}_{L\times NL}&\sum_{n=1}^N \mathbf{E}_{n}
\end{bmatrix}
\Big)-\operatorname{rank}\Big(
\begin{bmatrix}
\mathbf{I}_{NL}&\mathbf{E}\\
\end{bmatrix}\Big) \notag\\
&\stackrel{(b)}{=}&\operatorname{rank}\Big(
\begin{bmatrix}
\mathbf{I}_{NL}&\mathbf{0}_{NL\times \frac{N(N-1)}{N-T}L}\\
\mathbf{0}_{L\times NL}&\sum_{n=1}^N \mathbf{E}_{n}
\end{bmatrix}
\Big)-\operatorname{rank}\Big(
\begin{bmatrix}
\mathbf{I}_{NL}&\mathbf{0}_{NL\times \frac{N(N-1)}{N-T}L}\\
\end{bmatrix}\Big) \notag\\
&=&\operatorname{rank}\Big(\mathbf{I}_{NL}\Big)+\operatorname{rank}\Big(\sum_{n=1}^N \mathbf{E}_{n}\Big)-\operatorname{rank}\Big(\mathbf{I}_{NL}\Big) \notag\\
&=&\operatorname{rank}\Big(\sum_{n=1}^N \mathbf{E}_{n}\Big),\nonumber
\end{IEEEeqnarray}
where $(a)$ follows from the fact that the data $\{W_n,Z_n\}_{n\in[N]}$ are independently and uniformly distributed over the finite field $\mathbb{F}_q$, and $(b)$ is due to the fact that elementary column operations do not change the rank of a matrix. 
Therefore, for the coding framework to satisfy the correctness constraint in \eqref{model:correctness}, the condition in \eqref{eq:zerosum-constraint} must hold.

 
Next, we prove that the condition in \eqref{eq:rankcondition} is sufficient and necessary for the coding framework to satisfy the security constraint in \eqref{model:security}. In this proof, the condition in \eqref{eq:zerosum-constraint} can be invoked, since it has already been shown to be sufficient and necessary for the coding framework to satisfy the correctness constraint in \eqref{model:correctness}.

For any subset of colluding users $\mathcal{T}\subseteq[N]$ with size $|\mathcal{T}|\leq T$ and its complement $\mathcal{T}^c=[N]\backslash\mathcal{T}$ with size $|\mathcal{T}^c|=N-|\mathcal{T}|$, without loss of generality, we index $\mathcal{T}^c$ as the ordered set $\{n_1,n_2,\dots,n_{N-|\mathcal{T}|}\}$ with $n_1<n_2<\ldots<n_{N-|\mathcal{T}|}$, and $\mathcal{T}$ as the ordered set $\{n_{N-|\mathcal{T}|+1},n_{N-|\mathcal{T}|+2},\dots,n_{N}\}$ with $n_{N-|\mathcal{T}|+1}<n_{N-|\mathcal{T}|+2}<\ldots<n_{N}$. Then, the constructed coding framework satisfies
\begin{IEEEeqnarray}{rCl}
&&I(\{W_n\}_{n\in[N]}; \{X_n\}_{n\in[N]}, \{Z_n\}_{n\in \mathcal{T}}, \{Z_{m,n}\}_{m\in[N]\backslash\{n\},n\in\mathcal{T}}|W_{sum},\{W_n\}_{n\in\mathcal{T}})\nonumber\\
&=&I(\{W_n\}_{n\in \mathcal{T}^c}; \{X_n\}_{n\in [N]}, \{Z_n\}_{n\in \mathcal{T}}, \{Z_{m,n}\}_{m\in[N]\backslash\{n\},n\in\mathcal{T}}|W_{sum},\{W_n\}_{n\in\mathcal{T}})\nonumber\\
&=&I(\{W_n\}_{n\in \mathcal{T}^c};\{Z_n\}_{n\in \mathcal{T}}, \{Z_{m,n}\}_{m\in[N]\backslash\{n\},n\in\mathcal{T}}|W_{sum},\{W_n\}_{n\in\mathcal{T}})\nonumber\\
&&+I(\{W_n\}_{n\in \mathcal{T}^c}; \{X_n\}_{n\in [N]}|W_{sum},\{W_n\}_{n\in\mathcal{T}},\{Z_n\}_{n\in \mathcal{T}}, \{Z_{m,n}\}_{m\in[N]\backslash\{n\},n\in\mathcal{T}})\nonumber\\
&\overset{(a)}{=}&I(\{W_n\}_{n\in \mathcal{T}^c}; \{X_n\}_{n\in [N]}|W_{sum},\{W_n\}_{n\in\mathcal{T}},\{Z_n\}_{n\in \mathcal{T}}, \{Z_{m,n}\}_{m\in[N]\backslash\{n\},n\in\mathcal{T}})\nonumber\\
&\stackrel{(b)}{=}&I(\{W_n\}_{n\in \mathcal{T}^c}; \{X_n\}_{n\in \mathcal{T}^c}|W_{sum},\{W_n\}_{n\in\mathcal{T}},\{Z_n\}_{n\in \mathcal{T}}, \{Z_{m,n}\}_{m\in[N]\backslash\{n\},n\in\mathcal{T}})\nonumber\\
&=&H(\{X_n\}_{n\in \mathcal{T}^c}|W_{sum},\{W_n\}_{n\in\mathcal{T}},\{Z_n\}_{n\in \mathcal{T}}, \{Z_{m,n}\}_{m\in[N]\backslash\{n\},n\in\mathcal{T}})\nonumber\\
&&-H(\{X_n\}_{n\in \mathcal{T}^c}|\{W_n\}_{n\in [N]},\{Z_n\}_{n\in \mathcal{T}}, \{Z_{m,n}\}_{m\in[N]\backslash\{n\},n\in\mathcal{T}}), 
\label{eq:security_proof}
\end{IEEEeqnarray}
where $(a)$ follows from \eqref{eq:matrixofcodedkey} and the independence between the random keys $Z_1,\ldots,Z_N$ and the inputs $W_1,\ldots,W_N$ such that
\begin{IEEEeqnarray}{rCl}
   0&=&I(W_1,\ldots,W_N;Z_1,\ldots,Z_N)\notag\\
   &\geq& I(\{W_n\}_{n\in \mathcal{T}^c};\{Z_n\}_{n\in \mathcal{T}}, \{Z_{m,n}\}_{m\in[N]\backslash\{n\},n\in\mathcal{T}}|W_{sum},\{W_n\}_{n\in\mathcal{T}}) \notag\\
   &\geq&0, \notag
\end{IEEEeqnarray}
and $(b)$ holds since $X_n$ is a deterministic function of $W_n,Z_n$, and $\{Z_{m,n}\}_{m\in[N]\backslash\{n\}}$ for any $n\in\mathcal{T}$ by \eqref{eq:maskinputlinear}.

Let $n_s$ be an element of the complement set $\mathcal{T}^c$ for some $s\in [N-|\mathcal{T}|]$, then we set $\widetilde{\mathcal{T}}^c=\mathcal{T}^c\backslash\{n_s\}$.
Therefore, the first term in \eqref{eq:security_proof} can be established as follows:
\begin{IEEEeqnarray}{rCl}
&&H\big(\{X_n\}_{n\in \mathcal{T}^c}|W_{sum},\{W_n\}_{n\in\mathcal{T}},\{Z_n\}_{n\in \mathcal{T}}, \{Z_{m,n}\}_{m\in[N]\backslash\{n\},n\in\mathcal{T}}\big) \nonumber\\
&\overset{(a)}{=}&H\big(\sum_{n \in \mathcal{T}^c}X_n,\{X_n\}_{n\in \widetilde{\mathcal{T}}^c}|W_{sum},\{W_n\}_{n\in\mathcal{T}},\{Z_n\}_{n\in \mathcal{T}}, \{Z_{m,n}\}_{m\in[N]\backslash\{n\},n\in\mathcal{T}}\big) \nonumber\\
&\overset{(b)}{=}&H\big(\sum_{n \in \mathcal{T}}X_n,\{X_n\}_{n\in \widetilde{\mathcal{T}}^c}|W_{sum},\{W_n\}_{n\in\mathcal{T}},\{Z_n\}_{n\in \mathcal{T}}, \{Z_{m,n}\}_{m\in[N]\backslash\{n\},n\in\mathcal{T}}\big) \nonumber\\
&\overset{(c)}{=}&H\big(\{X_n\}_{n\in \widetilde{\mathcal{T}}^c}|W_{sum},\{W_n\}_{n\in\mathcal{T}},\{Z_n\}_{n\in \mathcal{T}}, \{Z_{m,n}\}_{m\in[N]\backslash\{n\},n\in\mathcal{T}}\big) \nonumber\\
&\overset{(d)}{=}&H\big(\{W_n+\sum_{m \in [N]}\mathbf{E}_{m,n}Z_m\}_{n\in \widetilde{\mathcal{T}}^c}|\sum\limits_{n\in\mathcal{T}^c}W_n,\{W_n\}_{n\in\mathcal{T}},\{Z_n\}_{n\in \mathcal{T}}, \{Z_{m,n}\}_{m\in[N]\backslash\{n\},n\in\mathcal{T}}\big)\nonumber\\
&\overset{(e)}{\leq}&H\big(\{W_n+\sum_{m \in [N]}\mathbf{E}_{m,n}Z_m\}_{n\in \widetilde{\mathcal{T}}^c}|\{Z_n\}_{n\in \mathcal{T}}, \{Z_{m,n}\}_{m\in[N]\backslash\{n\},n\in\mathcal{T}}\big)\label{inequality:1111}\\
&\overset{(f)}{=}&H\big(\{W_n+\sum_{m \in [N]}\mathbf{E}_{m,n}Z_m\}_{n\in \widetilde{\mathcal{T}}^c}|\{Z_n\}_{n\in \mathcal{T}}, \{\mathbf{G}_{m,n}Z_m\}_{m\in[N]\backslash\{n\},n\in\mathcal{T}}\big)\nonumber\\
&=&H\big(\{W_n+\sum_{m \in \mathcal{T}^c}\mathbf{E}_{m,n}Z_m\}_{n\in \widetilde{\mathcal{T}}^c}|\{Z_n\}_{n\in \mathcal{T}}, \{\mathbf{G}_{m,n}Z_m\}_{m\in\mathcal{T}^c,n\in\mathcal{T}}\big)\label{first:term:1}\\
&\stackrel{(g)}{=}&H\big(\{W_n+\sum_{m \in \mathcal{T}^c}\mathbf{E}_{m,n}Z_m\}_{n\in \widetilde{\mathcal{T}}^c}| \{\mathbf{G}_{m,n}Z_{m}\}_{m\in\mathcal{T}^c,n\in\mathcal{T}}\big)\label{first:term:2}\\
&=&H\big(\{W_n+\sum_{m \in \mathcal{T}^c}\mathbf{E}_{m,n}Z_m\}_{n\in \widetilde{\mathcal{T}}^c},\{\mathbf{G}_{m,n}Z_{m}\}_{m\in\mathcal{T}^c,n\in\mathcal{T}}\big)-H\big(\{\mathbf{G}_{m,n}Z_{m}\}_{m\in\mathcal{T}^c,n\in\mathcal{T}}\big), \label{eq:security_proof_1}
\end{IEEEeqnarray}
where $(a)$ follows because the collection $\{\sum_{n \in \mathcal{T}^c}X_n,\{X_n\}_{n\in \widetilde{\mathcal{T}}^c}\}$ is equivalent to $\{X_n\}_{n\in\mathcal{T}^c}$;
$(b)$ is due to the fact that $\sum_{n \in[N]}X_n=W_{sum}$ by \eqref{eq:zerosum-constraint} and hence $\sum_{n \in \mathcal{T}^c}X_n=W_{sum}-\sum_{n \in \mathcal{T}}X_n$;
$(c)$ follows from the fact that $X_n$ is determined by $W_n,Z_n$, and $\{Z_{m,n}\}_{m\in[N]\backslash\{n\}}$ for any $n\in\mathcal{T}$ by \eqref{eq:maskinputlinear};
$(d)$ follows from \eqref{eq:encodingscheme};
$(e)$ holds because conditioning cannot increase entropy;
$(f)$ is due to \eqref{eq:matrixofcodedkey};
$(g)$ follows from the fact that the random keys $\{Z_n\}_{n\in\mathcal{T}}$ are generated independently of $\{Z_m\}_{m\in\mathcal{T}^c}$ and $\{W_n\}_{n\in[N]}$, and hence are independent of $\{\mathbf{G}_{m,n}Z_m\}_{m\in\mathcal{T}^c,n\in\mathcal{T}}$ and $\{W_n+\sum_{m \in\mathcal{T}^c}\mathbf{E}_{m,n}Z_m\}_{n\in \widetilde{\mathcal{T}}^c}$.
Note that the inequality in \eqref{inequality:1111} holds with equality when $\{W_{n}\}_{n\in[N]}$ are independently and uniformly distributed. This is because, under this assumption, $\sum_{n\in\mathcal{T}^c}W_n$ and $\{W_n\}_{n\in\mathcal{T}}$ are independent of the collection $\{\{W_n+\sum_{m \in [N]}\mathbf{E}_{m,n}Z_m\}_{n\in \widetilde{\mathcal{T}}^c},\{Z_n\}_{n\in \mathcal{T}}, \{Z_{m,n}\}_{m\in[N]\backslash\{n\},n\in\mathcal{T}}\}$.

Similarly, the second term in \eqref{eq:security_proof} can be further analyzed as follows:
\begin{IEEEeqnarray}{rCl}
&&H\big(\{X_n\}_{n\in \mathcal{T}^c}|\{W_n\}_{n\in [N]},\{Z_n\}_{n\in \mathcal{T}}, \{Z_{m,n}\}_{m\in[N]\backslash\{n\},n\in\mathcal{T}}\big) \nonumber\\
&\overset{(a)}{=}&H\big(\{\sum_{m\in\mathcal{T}^c} \mathbf{E}_{m,n}Z_m\}_{n\in\mathcal{T}^c}|\{W_n\}_{n\in [N]},\{Z_n\}_{n\in \mathcal{T}}, \{\mathbf{G}_{m,n}Z_m\}_{m\in[N]\backslash\{n\},n\in\mathcal{T}}\big)\nonumber\\
&\overset{(b)}{=}&H\big(\{\sum_{m\in\mathcal{T}^c} \mathbf{E}_{m,n}Z_m\}_{n\in\mathcal{T}^c}|\{Z_n\}_{n\in \mathcal{T}}, \{\mathbf{G}_{m,n}Z_m\}_{m\in[N]\backslash\{n\},n\in\mathcal{T}}\big)\nonumber\\
&\overset{(c)}{=}&H\big(\{\sum_{m\in\mathcal{T}^c} \mathbf{E}_{m,n}Z_m\}_{n\in\mathcal{T}^c}|\{\mathbf{G}_{m,n}Z_{m}\}_{m\in\mathcal{T}^c,n\in\mathcal{T}}\big)\nonumber\\
&\overset{(d)}{=}&H\big(\{\sum_{m\in\mathcal{T}^c} \mathbf{E}_{m,n}Z_m\}_{n\in\mathcal{T}^c},\{\sum_{m\in\mathcal{T}^c} \mathbf{E}_{m,n}Z_m\}_{n\in\mathcal{T}}|\{\mathbf{G}_{m,n}Z_{m}\}_{m\in\mathcal{T}^c,n\in\mathcal{T}}\big)\nonumber\\
&\overset{(e)}{=}&H\big(\{\sum_{m\in\mathcal{T}^c} \mathbf{E}_{m,n}Z_m\}_{n\in\widetilde{\mathcal{T}}^c},\{\sum_{m\in\mathcal{T}^c} \mathbf{E}_{m,n}Z_m\}_{n\in\mathcal{T}}|\{\mathbf{G}_{m,n}Z_{m}\}_{m\in\mathcal{T}^c,n\in\mathcal{T}}\big)\nonumber\\
&\overset{(f)}{=}&H\big(\{\sum_{m\in\mathcal{T}^c} \mathbf{E}_{m,n}Z_m\}_{n\in\widetilde{\mathcal{T}}^c}|\{\mathbf{G}_{m,n}Z_{m}\}_{m\in\mathcal{T}^c,n\in\mathcal{T}}\big)\nonumber\\
&=&H\big(\{\sum_{m\in\mathcal{T}^c} \mathbf{E}_{m,n}Z_m\}_{n\in\widetilde{\mathcal{T}}^c},\{\mathbf{G}_{m,n}Z_{m}\}_{m\in\mathcal{T}^c,n\in\mathcal{T}}\big)-H\big(\{\mathbf{G}_{m,n}Z_{m}\}_{m\in\mathcal{T}^c,n\in\mathcal{T}}\big),\label{eq:security_proof_2}
\end{IEEEeqnarray}
where $(a)$ is due to \eqref{eq:matrixofcodedkey} and \eqref{eq:encodingscheme};
$(b)$ follows from the fact that the inputs $\{W_n\}_{n\in[N]}$ are independent of the random keys $\{Z_n\}_{n\in[N]}$; $(c)$ follows from an argument similar to that in \eqref{first:term:1}-\eqref{first:term:2};
$(d)$ and $(f)$ hold because $\{\sum_{m\in\mathcal{T}^c} \mathbf{E}_{m,n}Z_m\}_{n\in\mathcal{T}}$ is a deterministic function of $\{\mathbf{G}_{m,n}Z_{m}\}_{m\in\mathcal{T}^c,n\in\mathcal{T}}$ by \eqref{framework:matrix:equation};
$(e)$ follows from the fact that $\sum_{n\in[N]}\sum_{m\in\mathcal{T}^c} \mathbf{E}_{m,n}Z_m=\mathbf{0}_{L\times 1}$ by \eqref{eq:zerosum-constraint} and therefore $\sum_{m\in\mathcal{T}^c} \mathbf{E}_{m,n_s}Z_m=\mathbf{0}_{L\times 1}-\sum_{n\in\mathcal{T}}\sum_{m\in\mathcal{T}^c} \mathbf{E}_{m,n}Z_m--\sum_{n\in\widetilde{\mathcal{T}}^c}\sum_{m\in\mathcal{T}^c} \mathbf{E}_{m,n}Z_m$.

By combining \eqref{eq:security_proof}, \eqref{eq:security_proof_1}, and \eqref{eq:security_proof_2}, we have
\begin{IEEEeqnarray}{rCl}
&&I\big(\{W_n\}_{n\in[N]}; \{X_n\}_{n\in[N]}, \{Z_n\}_{n\in \mathcal{T}}, \{Z_{m,n}\}_{m\in[N]\backslash\{n\},n\in\mathcal{T}}|W_{sum},\{W_n\}_{n\in\mathcal{T}}\big)\nonumber\\
&\leq& H\big(\{W_n+\sum_{m \in \mathcal{T}^c}\mathbf{E}_{m,n}Z_m\}_{n\in \widetilde{\mathcal{T}}^c},\{\mathbf{G}_{m,n}Z_{m}\}_{m\in\mathcal{T}^c,n\in\mathcal{T}}\big)\nonumber\\
&&-H\big(\{\sum_{m\in\mathcal{T}^c} \mathbf{E}_{m,n}Z_m\}_{n\in\widetilde{\mathcal{T}}^c},\{\mathbf{G}_{m,n}Z_{m}\}_{m\in\mathcal{T}^c,n\in\mathcal{T}}\big) \nonumber\\
&=& H\Big(
\begin{bmatrix}
\mathbf{I}_{(N-|\mathcal{T}|-1)L}&\mathbf{E}_{\widetilde{\mathcal{T}}^c,\mathcal{T}^c}\\
\mathbf{0}_{\frac{|\mathcal{T}|(N-|\mathcal{T}|)}{N-T}L\times(N-|\mathcal{T}|-1)L}&\mathrm{diag}(\mathbf{G}_{\mathcal{T}^c,\mathcal{T}})
\end{bmatrix}
\begin{bmatrix}
W_{n_1}\\
\vdots\\
W_{n_{s-1}}\\
W_{n_{s+1}}\\
\vdots\\
W_{n_{N-|\mathcal{T}|}}\\
Z_{n_1}\\
\vdots\\
Z_{n_{N-|\mathcal{T}|}}\\
\end{bmatrix}
\Big)-H\Big(
 \begin{bmatrix}
\mathbf{E}_{\widetilde{\mathcal{T}}^c,\mathcal{T}^c}\\
\mathrm{diag}(\mathbf{G}_{\mathcal{T}^c,\mathcal{T}})
 \end{bmatrix}
\begin{bmatrix}
Z_{n_1}\\
\vdots\\
Z_{n_{N-|\mathcal{T}|}}\\
\end{bmatrix} \Big) \nonumber \\
&\overset{(a)}{\leq}&\operatorname{rank}\Big(
\begin{bmatrix}
\mathbf{I}_{(N-|\mathcal{T}|-1)L}&\mathbf{E}_{\widetilde{\mathcal{T}}^c,\mathcal{T}^c}\\
\mathbf{0}_{\frac{|\mathcal{T}|(N-|\mathcal{T}|)}{N-T}L\times(N-|\mathcal{T}|-1)L}&\mathrm{diag}(\mathbf{G}_{\mathcal{T}^c,\mathcal{T}})
\end{bmatrix}
\Big)-\operatorname{rank}\Big(
 \begin{bmatrix}
\mathbf{E}_{\widetilde{\mathcal{T}}^c,\mathcal{T}^c}\\
\mathrm{diag}(\mathbf{G}_{\mathcal{T}^c,\mathcal{T}})
 \end{bmatrix}
\Big)\label{inequality:2222} \\
&=&\operatorname{rank}\Big(
\begin{bmatrix}
\mathbf{I}_{(N-|\mathcal{T}|-1)L}&\mathbf{0}_{(N-|\mathcal{T}|-1)L\times\frac{(N-|\mathcal{T}|)(N-1)}{N-T}L}\\
\mathbf{0}_{\frac{|\mathcal{T}|(N-|\mathcal{T}|)}{N-T}L\times(N-|\mathcal{T}|-1)L}&\mathrm{diag}(\mathbf{G}_{\mathcal{T}^c,\mathcal{T}})
\end{bmatrix}
\Big)-\operatorname{rank}\Big(
 \begin{bmatrix}
\mathbf{E}_{\widetilde{\mathcal{T}}^c,\mathcal{T}^c}\\
\mathrm{diag}(\mathbf{G}_{\mathcal{T}^c,\mathcal{T}})
 \end{bmatrix}
\Big)\nonumber \\
&=&(N-|\mathcal{T}|-1)L+\operatorname{rank}\Big(\mathrm{diag}(\mathbf{G}_{\mathcal{T}^c,\mathcal{T}})\Big)-\operatorname{rank}\Big(
 \begin{bmatrix}
\mathbf{E}_{\widetilde{\mathcal{T}}^c,\mathcal{T}^c}\\
\mathrm{diag}(\mathbf{G}_{\mathcal{T}^c,\mathcal{T}})
 \end{bmatrix}
\Big), \notag
\end{IEEEeqnarray}
where $(a)$ follows from the fact that the random keys $Z_{n_1},\ldots,Z_{n_{N-|\mathcal{T}|}}$ are independently and uniformly distributed over the finite field $\mathbb{F}_q$. 
Note that the inequality in \eqref{inequality:2222} holds with equality when $\{W_{n}\}_{n\in[N]}$ are independently and uniformly distributed. 

Consequently, if the encoding matrices satisfy the conditions in \eqref{eq:zerosum-constraint} and \eqref{eq:rankcondition}, then the constructed coding framework satisfies
\begin{IEEEeqnarray}{rCl}
0&\leq&I\big(\{W_n\}_{n\in[N]}; \{X_n\}_{n\in[N]}, \{Z_n\}_{n\in \mathcal{T}}, \{Z_{m,n}\}_{m\in[N]\backslash\{n\},n\in\mathcal{T}}|W_{sum},\{W_n\}_{n\in\mathcal{T}}\big) \notag \\
&\leq&(N-|\mathcal{T}|-1)L+\operatorname{rank}\Big(\mathrm{diag}(\mathbf{G}_{\mathcal{T}^c,\mathcal{T}})\Big)-\operatorname{rank}\Big(
 \begin{bmatrix}
\mathbf{E}_{\widetilde{\mathcal{T}}^c,\mathcal{T}^c}\\
\mathrm{diag}(\mathbf{G}_{\mathcal{T}^c,\mathcal{T}})
 \end{bmatrix}
\Big) \notag \\
&=&0, \notag
\end{IEEEeqnarray}
which matches the security constraint in \eqref{model:security}. Therefore, the conditions in \eqref{eq:zerosum-constraint} and \eqref{eq:rankcondition} are sufficient for the coding framework to simultaneously satisfy the correctness constraint in \eqref{model:correctness} and the security constraint in \eqref{model:security}.
Furthermore, as shown in the proof, when the inputs $\{W_{n}\}_{n\in[N]}$ are independently and uniformly distributed over $\mathbb{F}_q$, the inequalities in \eqref{inequality:1111} and \eqref{inequality:2222} hold with equality. Therefore, we obtain the following stronger result:
\begin{IEEEeqnarray*}{rCl}
&&I\big(\{W_n\}_{n\in[N]}; \{X_n\}_{n\in[N]}, \{Z_n\}_{n\in \mathcal{T}}, \{Z_{m,n}\}_{m\in[N]\backslash\{n\},n\in\mathcal{T}}|W_{sum},\{W_n\}_{n\in\mathcal{T}}\big) \notag \\
&=&(N-|\mathcal{T}|-1)L+\operatorname{rank}\Big(\mathrm{diag}(\mathbf{G}_{\mathcal{T}^c,\mathcal{T}})\Big)-\operatorname{rank}\Big(
 \begin{bmatrix}
\mathbf{E}_{\widetilde{\mathcal{T}}^c,\mathcal{T}^c}\\
\mathrm{diag}(\mathbf{G}_{\mathcal{T}^c,\mathcal{T}})
 \end{bmatrix}
\Big).
\end{IEEEeqnarray*}
Accordingly, for any achievable coding framework that satisfies both the correctness constraint in \eqref{model:correctness} and the security constraint in \eqref{model:security}, the condition in \eqref{eq:rankcondition} must hold. This completes the proof of the proposition.

\end{appendix}

\bibliographystyle{ieeetr}
\bibliography{reference.bib}

@inproceedings{mcmahan2017communication,
  title={Communication-efficient learning of deep networks from decentralized data},
  author={McMahan, Brendan and Moore, Eider and Ramage, Daniel and Hampson, Seth and y Arcas, Blaise Aguera},
  booktitle={Proc. 20th Int. Conf. Artif. Intell. Stat.},
  pages={1273--1282},
  month = oct,
  year={2017},
}

@article{konevcny2016federated,
  title={Federated learning: Strategies for improving communication efficiency},
  author={Kone{\v{c}}n{\`y}, Jakub and McMahan, H Brendan and Yu, Felix X and Richt{\'a}rik, Peter and Suresh, Ananda Theertha and Bacon, Dave},
  journal={arXiv preprint arXiv:1610.05492},
  year={2016}
}

@article{li2020federated,
  title={Federated learning: Challenges, methods, and future directions},
  author={Li, Tian and Sahu, Anit Kumar and Talwalkar, Ameet and Smith, Virginia},
  journal={IEEE signal processing magazine},
  volume={37},
  number={3},
  pages={50--60},
  year={2020},
  publisher={IEEE}
}

@inproceedings{fredrikson2015model,
  title={Model inversion attacks that exploit confidence information and basic countermeasures},
  author={Fredrikson, Matt and Jha, Somesh and Ristenpart, Thomas},
  booktitle={Proceedings of the 22nd ACM SIGSAC conference on computer and communications security},
  pages={1322--1333},
  year={2015}
}

@inproceedings{bonawitz2017practical,
  title={Practical secure aggregation for privacy-preserving machine learning},
  author={Bonawitz, Keith and Ivanov, Vladimir and Kreuter, Ben and Marcedone, Antonio and McMahan, H Brendan and Patel, Sarvar and Ramage, Daniel and Segal, Aaron and Seth, Karn},
  booktitle={proceedings of the 2017 ACM SIGSAC Conference on Computer and Communications Security},
  pages={1175--1191},
  year={2017}
}

@article{zhao2022information,
  title={Information theoretic secure aggregation with user dropouts},
  author={Zhao, Yizhou and Sun, Hua},
  journal={IEEE Transactions on Information Theory},
  volume={68},
  number={11},
  pages={7471--7484},
  year={2022},
  publisher={IEEE}
}

@article{wan2024information,
  title={On the information theoretic secure aggregation with uncoded groupwise keys},
  author={Wan, Kai and Yao, Xin and Sun, Hua and Ji, Mingyue and Caire, Giuseppe},
  journal={IEEE Transactions on Information Theory},
  volume={70},
  number={9},
  pages={6596--6619},
  year={2024},
  publisher={IEEE}
}

@book{katz2007introduction,
  title={Introduction to modern cryptography: principles and protocols},
  author={Katz, Jonathan and Lindell, Yehuda},
  year={2007},
  publisher={Chapman and hall/CRC}
}

@article{zhao2023secure,
  title={Secure summation: Capacity region, groupwise key, and feasibility},
  author={Zhao, Yizhou and Sun, Hua},
  journal={IEEE Transactions on Information Theory},
  volume={70},
  number={2},
  pages={1376--1387},
  year={2023},
  publisher={IEEE}
}

@book{horn1994topics,
  title={Topics in matrix analysis},
  author={Horn, Roger A and Johnson, Charles R},
  year={1994},
  publisher={Cambridge university press}
}

@article{yang2019federated,
  title={Federated machine learning: Concept and applications},
  author={Yang, Qiang and Liu, Yang and Chen, Tianjian and Tong, Yongxin},
  journal={ACM Transactions on Intelligent Systems and Technology (TIST)},
  volume={10},
  number={2},
  pages={1--19},
  year={2019},
  publisher={ACM New York, NY, USA}
}

@inproceedings{yuan2025vector,
  title={Vector linear secure aggregation},
  author={Yuan, Xihang and Sun, Hua},
  booktitle={2025 IEEE International Symposium on Information Theory (ISIT)},
  pages={1--6},
  year={2025},
  organization={IEEE}
}

@article{li2025weakly,
  author={Li, Zhou and Zhao, Yizhou and Sun, Hua},
  journal={IEEE Transactions on Information Theory}, 
  title={Weakly Secure Summation With Colluding Users}, 
  year={2025},
  volume={71},
  number={7},
  pages={5672-5683}}

@article{zhang2025optimal,
  title={Optimal communication and key rate region for hierarchical secure aggregation with user collusion},
  author={Zhang, Xiang and Wan, Kai and Sun, Hua and Wang, Shiqiang and Ji, Mingyue and Caire, Giuseppe},
  journal={IEEE Transactions on Information Theory},
  volume={72},
  number={2},
  pages={1030--1050},
  year={2025},
  publisher={IEEE}
}

@article{li2025capacity,
  title={The capacity of collusion-resilient decentralized secure aggregation with groupwise keys},
  author={Li, Zhou and Zhang, Xiang and Zhao, Yizhou and Chen, Haiqiang and Fan, Jihao and Caire, Giuseppe},
  journal={arXiv preprint arXiv:2511.14444},
  year={2025}
}

@inproceedings{bell2020secure,
  title={Secure single-server aggregation with (poly) logarithmic overhead},
  author={Bell, James Henry and Bonawitz, Kallista A and Gasc{\'o}n, Adri{\`a} and Lepoint, Tancr{\`e}de and Raykova, Mariana},
  booktitle={Proceedings of the 2020 ACM SIGSAC conference on computer and communications security},
  pages={1253--1269},
  year={2020}
}

@article{choi2020communication,
  title={Communication-computation efficient secure aggregation for federated learning},
  author={Choi, Beongjun and Sohn, Jy-yong and Han, Dong-Jun and Moon, Jaekyun},
  journal={arXiv preprint arXiv:2012.05433},
  year={2020}
}

@article{zheng2022aggregation,
  title={Aggregation service for federated learning: An efficient, secure, and more resilient realization},
  author={Zheng, Yifeng and Lai, Shangqi and Liu, Yi and Yuan, Xingliang and Yi, Xun and Wang, Cong},
  journal={IEEE Transactions on Dependable and Secure Computing},
  volume={20},
  number={2},
  pages={988--1001},
  year={2022},
  publisher={IEEE}
}

@article{zhang2025securegroup,
  title={On Secure Aggregation With Uncoded Groupwise Keys Against User Dropouts and User Collusion},
  author={Zhang, Ziting and Liu, Jiayu and Wan, Kai and Sun, Hua and Ji, Mingyue and Caire, Giuseppe},
  journal={IEEE Transactions on Information Theory},
  volume={71},
  number={11},
  pages={8391--8413},
  year={2025}
}

@article{hu2026capacity,
  title={On the Capacity Region of Individual Key Rates in Vector Linear Secure Aggregation},
  author={Hu, Lei and Ulukus, Sennur},
  journal={arXiv preprint arXiv:2601.03241},
  year={2026}
}

@article{Zhang2026Decentralized,
  author={Zhang, Xiang and Li, Zhou and Li, Shuangyang and Wan, Kai and Ng, Derrick Wing Kwan and Caire, Giuseppe},
  journal={IEEE Journal on Selected Areas in Communications}, 
  title={Information-Theoretic Decentralized Secure Aggregation With Passive Collusion Resilience}, 
  year={2026},
  volume={44},
  number={},
  pages={4414-4428}}

@article{li2025hierarchical,
  title={Hierarchical Secure Aggregation with Heterogeneous Security Constraints and Arbitrary User Collusion},
  author={Li, Zhou and Zhang, Xiang and Lv, Jiawen and Fan, Jihao and Chen, Haiqiang and Caire, Giuseppe},
  journal={arXiv preprint arXiv:2507.14768},
  year={2025}
}

@inproceedings{zhu2019deep,
  title={Deep leakage from gradients},
  author={Zhu, Ligeng and Liu, Zhijian and Han, Song},
  booktitle={Advances in neural information processing systems},
  pages= {14747--14756},
  year={2019}
}

@inproceedings{geiping2020inverting,
  title={Inverting gradients-how easy is it to break privacy in federated learning?},
  author={Geiping, Jonas and Bauermeister, Hartmut and Dr{\"o}ge, Hannah and Moeller, Michael},
  booktitle={Advances in neural information processing systems},
  pages={16937--16947},
  year={2020}
}

@article{wan2024capacity,
  title={The capacity region of information theoretic secure aggregation with uncoded groupwise keys},
  author={Wan, Kai and Sun, Hua and Ji, Mingyue and Mi, Tiebin and Caire, Giuseppe},
  journal={IEEE Transactions on Information Theory},
  volume={70},
  number={10},
  pages={6932--6949},
  year={2024},
  publisher={IEEE}
}

@article{shamir1979share,
  title={How to share a secret},
  author={Shamir, Adi},
  journal={Communications of the ACM},
  volume={22},
  number={11},
  pages={612--613},
  year={1979},
  publisher={ACm New York, NY, USA}
}

@article{li2026fundamental,
  title={On the Fundamental Limits of Hierarchical Secure Aggregation with Dropout and Collusion Resilience},
  author={Li, Zhou and Zhao, Yizhou and Zhang, Xiang and Caire, Giuseppe},
  journal={arXiv preprint arXiv:2603.19705},
  year={2026}
}

@article{scarani2009security,
  title={The security of practical quantum key distribution},
  author={Scarani, Valerio and Bechmann-Pasquinucci, Helle and Cerf, Nicolas J and Du{\v{s}}ek, Miloslav and L{\"u}tkenhaus, Norbert and Peev, Momtchil},
  journal={Reviews of modern physics},
  volume={81},
  number={3},
  pages={1301--1350},
  year={2009},
  publisher={APS}
}

@article{liu2022efficient,
  title={Efficient dropout-resilient aggregation for privacy-preserving machine learning},
  author={Liu, Ziyao and Guo, Jiale and Lam, Kwok-Yan and Zhao, Jun},
  journal={IEEE Transactions on Information Forensics and Security},
  volume={18},
  pages={1839--1854},
  year={2022},
  publisher={IEEE}
}

@article{9830997,
  author={Liu, Ziyao and Guo, Jiale and Yang, Wenzhuo and Fan, Jiani and Lam, Kwok-Yan and Zhao, Jun},
  journal={IEEE Transactions on Big Data}, 
  title={Privacy-Preserving Aggregation in Federated Learning: A Survey}, 
  year={2022},
  volume={},
  number={},
  pages={1-20},
  doi={10.1109/TBDATA.2022.3190835}}

@article{kadhe2020fastsecagg,
  title={Fastsecagg: Scalable secure aggregation for privacy-preserving federated learning},
  author={Kadhe, Swanand and Rajaraman, Nived and Koyluoglu, O Ozan and Ramchandran, Kannan},
  journal={arXiv preprint arXiv:2009.11248},
  year={2020}
}

@article{so2021turbo,
  title={Turbo-aggregate: Breaking the quadratic aggregation barrier in secure federated learning},
  author={So, Jinhyun and G{\"u}ler, Ba{\c{s}}ak and Avestimehr, A Salman},
  journal={IEEE Journal on Selected Areas in Information Theory},
  volume={2},
  number={1},
  pages={479--489},
  year={2021},
  publisher={IEEE}
}

@article{jahani2023swiftagg+,
  title={SwiftAgg+: Achieving asymptotically optimal communication loads in secure aggregation for federated learning},
  author={Jahani-Nezhad, Tayyebeh and Maddah-Ali, Mohammad Ali and Li, Songze and Caire, Giuseppe},
  journal={IEEE Journal on Selected Areas in Communications},
  volume={41},
  number={4},
  pages={977--989},
  year={2023},
  publisher={IEEE}
}

@article{zhang2024survey,
  title={A survey of trustworthy federated learning: Issues, solutions, and challenges},
  author={Zhang, Yifei and Zeng, Dun and Luo, Jinglong and Fu, Xinyu and Chen, Guanzhong and Xu, Zenglin and King, Irwin},
  journal={ACM Transactions on Intelligent Systems and Technology},
  volume={15},
  number={6},
  pages={1--47},
  year={2024},
  publisher={ACM New York, NY, USA}
}

@inproceedings{bonawitz2019federated,
  title={Federated learning with autotuned communication-efficient secure aggregation},
  author={Bonawitz, Keith and Salehi, Fariborz and Kone{\v{c}}n{\`y}, Jakub and McMahan, Brendan and Gruteser, Marco},
  booktitle={2019 53rd Asilomar Conference on Signals, Systems, and Computers},
  pages={1222--1226},
  year={2019},
  organization={IEEE}
}

@inproceedings{so2022lightsecagg,
  title={Lightsecagg: a lightweight and versatile design for secure aggregation in federated learning},
  author={So, Jinhyun and He, Chaoyang and Yang, Chien-Sheng and Li, Songze and Yu, Qian and E Ali, Ramy and Guler, Basak and Avestimehr, Salman},
   booktitle={Proceedings of Machine Learning and Systems},
  pages={694--720},
  year={2022}
}

@article{wen2023survey,
  title={A survey on federated learning: challenges and applications},
  author={Wen, Jie and Zhang, Zhixia and Lan, Yang and Cui, Zhihua and Cai, Jianghui and Zhang, Wensheng},
  journal={International journal of machine learning and cybernetics},
  volume={14},
  number={2},
  pages={513--535},
  year={2023},
  publisher={Springer}
}

@article{lu2023top,
  title={Top-k sparsification with secure aggregation for privacy-preserving federated learning},
  author={Lu, Shiwei and Li, Ruihu and Liu, Wenbin and Guan, Chaofeng and Yang, Xiaopeng},
  journal={Computers \& Security},
  volume={124},
  pages={102993},
  year={2023},
  publisher={Elsevier}
}

@article{ergun2021sparsified,
  title={Sparsified secure aggregation for privacy-preserving federated learning},
  author={Ergun, Irem and Sami, Hasin Us and Guler, Basak},
  journal={arXiv preprint arXiv:2112.12872},
  year={2021}
}
\end{document}